\documentclass[12pt,reqno]{article}

\usepackage[T1]{fontenc}
\usepackage[utf8]{inputenc}
\usepackage{lmodern}
\usepackage{microtype}

\usepackage[a4paper,margin=1in]{geometry}
\usepackage{setspace}
\usepackage{natbib}
\usepackage{amsmath,amssymb,amsthm,mathtools,bm}
\numberwithin{equation}{section}

\usepackage{graphicx}
\usepackage{subcaption}
\usepackage{booktabs}
\usepackage{multirow}
\usepackage{tabularx}
\usepackage{colortbl}
\usepackage{float}
\usepackage{comment}
\usepackage{tikz}
\usetikzlibrary{arrows.meta,calc,positioning}
\usepackage{pgfplots}
\pgfplotsset{compat=1.18}

\usepackage{enumitem}

\usepackage{lscape}

\usepackage{natbib}
\usepackage{hyperref}
\hypersetup{
  colorlinks=true,
  linkcolor=blue!50!black,
  citecolor=blue!50!black,
  urlcolor=blue!50!black,
  pdfauthor={},
  pdftitle={Partial Identification with Auxiliary Moment Restrictions}
}
\usepackage[nameinlink,capitalise,noabbrev]{cleveref}

\theoremstyle{plain}
\newtheorem{theorem}{Theorem}[section]
\newtheorem{proposition}[theorem]{Proposition}
\newtheorem{lemma}[theorem]{Lemma}
\newtheorem{corollary}[theorem]{Corollary}

\theoremstyle{definition}
\newtheorem{assumption}{Assumption}

\newtheorem{definition}[theorem]{Definition}

\theoremstyle{remark}
\newtheorem{remark}[theorem]{Remark}

\newcommand{\R}{\mathbb{R}}

\newcommand{\E}{\mathbb{E}}

\newcommand{\PP}{\mathbf{P}}

\newcommand{\1}{\mathbf{1}}

\DeclareMathOperator{\Sel}{\mathbf{Sel}}
\title{\bf{Partial Identification with Auxiliary Moment Restrictions}}
\author{
  Arie Beresteanu\thanks{University of Pittsburgh, Pittsburgh, PA, USA. \texttt{arie@pitt.edu}}
  \and
  Behrooz Moosavi Ramezanzadeh\thanks{University of Pittsburgh, Pittsburgh, PA, USA. \texttt{behroozmoosavi@pitt.edu}}
}
\date{\today}

\begin{document}
\maketitle

\begin{abstract}
Partial identification is often set aside in practice because the identification regions it delivers are too wide to be useful, pushing researchers toward strong assumptions that buy point identification at the cost of credibility. We show that a source of information already sitting in most interval-valued datasets can fix this without adding any assumption at all. When an outcome is reported only as an interval---because a data custodian bracketed, top-coded, or formally privatized it to protect respondents---the same custodian typically continues to publish accurate population aggregates of that outcome, precisely because doing so does not compromise any individual record. We develop a framework for exploiting exactly this information: restricting the set of admissible completions of the data to those consistent with a known aggregate, rather than restricting the interval itself, and characterizing the sharp identification region that results for the best linear predictor. The restrictions we study behave in strikingly different ways---some collapse the region by a full dimension, others narrow it while leaving its shape intact. We characterize the geometric effect of each restriction and derive closed-form directional measures of identifying value for the mean and conditional-mean cases. An illustration using interval-valued wages from the Current Population Survey shows that the effect is far from marginal: modest auxiliary information recovers a substantial share of the identifying power usually thought to be lost once an outcome is coarsened.
\end{abstract}

\noindent\textbf{Credit authorship contribution statement: }

\textbf{Arie Beresteanu: } Writing - review \& editing, Writing - original draft, Methodology, Investigation, Conceptualization.

\textbf{Behrooz Moosavi Ramezanzadeh: } Writing - review \& editing, Writing - original draft, Methodology, Investigation, Conceptualization.

\noindent\textbf{Conflict of Interest:} The authors declare that they have no conflict of interest.

\section{Introduction}\label{sec:intro}

The literature on partial identification offers a principled way to learn from data under weak and credible assumptions. Yet it is adopted less often than its intellectual standing would warrant. A recurring complaint among practitioners is that the identified sets these methods deliver are large and, as a result, uninformative for policy. Confronted with wide bounds, researchers frequently impose strong behavioral or distributional assumptions that restore point identification but erode credibility---the very property that partial identification was designed to protect.

This paper develops a remedy that adds \emph{information} rather than assumptions, and it does so in a setting where the information arises naturally. In a growing number of applications the outcome of interest is reported only as an interval, not by accident, but \emph{by design}: to protect respondent confidentiality, the data custodian deliberately coarsens the true value $y^*$ into a published range $[y_L,y_U]$. Income is bracketed or top-coded in public-use survey files; sensitive variables are released only within bins; and modern formal-privacy systems perturb record-level data before release. At the same time, the very same custodian continues to publish accurate population-level summaries of the same latent outcome---official means, or the exact aggregates that a formal-privacy system releases as invariants---precisely because such aggregates pass disclosure review without exposing any individual record. The econometrician therefore confronts interval-valued microdata alongside known population aggregates of the same unobserved outcome. We show that this auxiliary information---a byproduct of the privacy regime rather than an external dataset that must be located and merged---can substantially shrink the identified set for a best linear predictor (BLP), without imposing any restriction on the data-generating process beyond the interval-containment condition $\PP(y_L\le y^*\le y_U)=1$.

We study the predictor associated with squared loss: the best linear predictor that solves the population linear projection problem for the conditional expectation of the latent outcome. Our analysis builds directly on \citet{BM2008}, who characterize the sharp identified set for the conditional-expectation BLP with interval-valued outcomes. We ask how much identifying power is contained in knowledge of (i) the unconditional mean of the latent outcome, (ii) a known moment of a transformation of the latent outcome, and (iii) a conditional mean given a subset of the covariates or given an external variable that is excluded from the model.

\paragraph{Overview of the analysis and contributions.}
The central device in the paper is a restricted selection set. We leave the observed random interval $Y=[y_L,y_U]$ unchanged and restrict only the set of admissible selections from that interval. Each candidate selection represents a possible realization of the latent outcome $y^*$, and auxiliary information removes those selections whose moments do not agree with the known aggregate. The resulting identified region is therefore the image, under the best-linear-predictor map, of a selection set restricted by the available population information. Nonemptiness of the relevant restricted selection sets is established in the companion note \citet{beresteanu2025noterestrictedselectionset}.

We use this framework to study several combinations of target parameters and auxiliary restrictions. The baseline target is the coefficient vector of the best linear predictor of $y^*$. When the unconditional mean $E[y^*]=\kappa$ is known, the identified region is the intersection of the unrestricted region with an affine hyperplane. Under a nondegeneracy condition, the restriction therefore lowers the affine dimension of the region by one. We then ask how the result changes when the target or the auxiliary information involves a transformation $f(y^*)$. If the target is the best linear predictor of $f(y^*)$, knowledge of $E[y^*]=\kappa$ restricts the admissible selections but does not impose a direct affine restriction on the transformed-outcome coefficients. If instead $E[f(y^*)]=\kappa_f$ is known, the transformed-outcome coefficients satisfy an affine hyperplane restriction. By contrast, when the target remains the best linear predictor of $y^*$ and $E[f(y^*)]=\kappa_f$ is used only as auxiliary information, the restriction generally narrows the identified region without reducing its dimension.

We next consider conditional auxiliary information. If the researcher knows
$$
E[y^* \mid x_1]=\kappa(x_1)
$$
for a subvector $x_1$ of the regressors, a population Frisch--Waugh--Lovell decomposition shows that the corresponding coefficient block is an affine function of the coefficients on the remaining regressors. The restriction therefore imposes $d_1+1$ exact affine restrictions and reduces the full identified set to an affine image of a lower-dimensional coefficient region. If instead the researcher knows
$$
E[y^* \mid v]=\kappa(v)
$$
for an external variable $v$ excluded from the predictor, no coefficient is generally pinned down by an affine equation. The restriction tightens identification by fixing how the available interval width must be allocated across values of $v$, and it improves on the pooled-mean restriction whenever that allocation differs from the one selected by the pooled optimization problem.

Beyond characterizing these geometric effects, we quantify the identifying value of the mean and conditional-mean restrictions. Directional bounds can be written as allocation problems in which the interval width $y_U-y_L$ is assigned across observations according to the score associated with a direction in coefficient space. A known mean fixes the total amount of width that may be allocated, while a known conditional mean fixes how that amount is divided across covariate cells or values of an external variable. This representation yields formulas for the loss in each directional support value and shows that a restriction matters only when it forces the allocation away from the one an unrestricted optimizer would choose.

Finally, we extend the restricted-selection framework to nonlinear and discontinuous transformations. For continuous transformations, we establish convexity and derive a Lagrangian support representation under suitable regularity conditions. For indicator transformations, which encode information such as a known distributional probability at a threshold, we provide a finite-sample sorting characterization. This characterization connects this paper to the quantile analysis of \citet{beresteanu-2021}, but the role of the quantile information is different here: it is used as auxiliary information to sharpen identification of a mean-based best linear predictor rather than as the target of the analysis itself.

\paragraph{Related literature.}
At its core, partial identification asks what can be learned about a parameter that is only set-identified, emphasizing credible assumptions; see \citet{Manski2003} for a monograph treatment and \citet{Tamer2010} for a survey. The foundational treatment of interval-valued regression itself is \citet{manskitamer}, who derive nonparametric bounds on a regression function when an outcome or regressor is observed only within an interval; we work throughout in their interval-containment setting but target the BLP coefficient rather than the conditional expectation function directly, following \citet{BM2008}. A central methodological theme is to translate credible restrictions on the data-generating process or on agents' behavior into inequality restrictions on the parameters, and to show that these inequalities are both necessary and sufficient, so that the identified set is sharp. We obtain sharpness through the random-set apparatus of \citet{Molchanov2005,BMM2011} and \citet{MolchanovMolinari2014}, and we study how auxiliary information on the mean of the latent outcome, or of a transformation of it, tightens the corresponding selection sets and hence the identified regions.

The motivating environment connects our analysis to the literature on confidentiality protection. Agencies have long limited disclosure by bracketing or top-coding sensitive variables, which is precisely what makes the outcome interval-valued; more recently, national statistical offices have adopted formal privacy, as in the U.S. Census Bureau's 2020 Disclosure Avoidance System, which protects record-level data through differential privacy while releasing selected aggregates exactly as invariants \citep{Abowd2022}. We take such released aggregates as given and ask what they identify, so that the auxiliary information in our results is a feature of the disclosure regime rather than an additional assumption.

Our use of external information is also related to, but distinct from, the data-combination literature. \citet{CrossManski2002} consider a setting in which $(y,x,z)$ are never jointly observed and the researcher knows the marginals $P(y\mid x)$ and $P(z\mid x)$ from two separate sources, which are then combined to bound $E(y\mid x,z)$. In our setting a single sample is observed in which the outcome is reported only as an interval, and the additional information takes the form of known moments of the unobserved outcome. Unlike in \citet{CrossManski2002}, the conditional expectation of $y$ given the covariates is already partially identified without the auxiliary information; that information is used here to tighten the bounds. Our approach is also distinct from methods that treat the bounds on the outcome as themselves estimated or as an unknown object to be restricted directly: \citet{ChandrasekharEtAl2019} construct BLP regions from an estimated band around the outcome and develop inference for the resulting support function, while \citet{MagnacMaurin2008} and \citet{beresteanu-2021} restrict an unknown bounded function or transformation of the outcome. We instead hold the observed interval $Y=[y_L,y_U]$ fixed throughout and restrict only the \emph{set of selections} admissible within it, so that auxiliary information enters through which completions of the observed interval are permitted rather than through the interval itself. Finally, shape restrictions such as monotonicity or concavity are typically imposed on selections by redefining the random set from which they are drawn; \citet[p.~65]{MolchanovMolinari2014} discuss restricting the values of selections through intersection with another (possibly random) set. We instead leave the observed random set intact and restrict the \emph{set of selections} to those whose moments match externally known values. In spirit, this is close to the use of monotone-instrument and other auxiliary restrictions to sharpen partially identified objects, as in \citet{ManskiPepper2000}; the present paper extends \citet{BM2008} by showing that auxiliary moments---of the outcome itself, of a subvector of covariates, of an external variable, and of transformations of the outcome, including quantiles---carry identifying power, and by quantifying it.

\paragraph{Roadmap.}
\Cref{sec:condexp} develops the results summarized above: the unconstrained region and its basic properties, the hyperplane characterization under a known unconditional mean, directional and coordinate bounds, the value of a known mean, identification with transformations (with and without retargeting), a known conditional mean given a subvector of covariates, and a known conditional mean given an external variable. \Cref{sec:numerical} illustrates every result numerically using interval-valued wages constructed from the 2020 March Current Population Survey (CPS) via a stylized coarsening mechanism that plays the role of a privacy device applied to the observed wage, in the spirit of the data used by \citet{CHT2007} and \citet{BM2008}. \Cref{app:numerical} reports supplementary numerical detail. \Cref{app:background} collects the background results from random-set theory and convex analysis used in the proofs, including the general treatment of vector-valued transformations and quantile information; \Cref{app:proofs} contains all proofs.

\section{Results}\label{sec:condexp}

In this section, we characterize the identification region for the parameter
vector $\theta$ of a best linear predictor under squared loss when the outcome
is interval-valued and the researcher has access to auxiliary population
information about the latent outcome.

Let $(\Omega,\mathcal F,\PP)$ be a probability space. Let
$y^*\in\R$ be a latent scalar outcome, and let
\[
Y=[y_L,y_U]
\]
be an observed random interval satisfying
\[
\PP(y_L\leq y^*\leq y_U)=1.
\]
Let $x\in\R^d$ be an observed vector of covariates, and define the augmented
covariate vector
\[
\tilde x=(1,x_1,\dots,x_d)'
\in\R^{d+1}.
\]

Throughout this section, we impose the following regularity conditions.

\begin{assumption}[Baseline regularity]\label{as:main}
The random variables satisfy
\[
\E\!\left[\|\tilde x\|^2\right]<\infty,
\qquad
\E\!\left[
\|\tilde x\|\bigl(|y_L|+|y_U|\bigr)
\right]<\infty,
\]
and the second-moment matrix
\[
Q:=\E[\tilde x\tilde x']
\]
is nonsingular.
\end{assumption}

Because $Q$ is a symmetric second-moment matrix, nonsingularity is equivalent
to positive definiteness.

Define the set of integrable measurable selections of $Y$ by
\[
\Sel^1(Y)
=
\left\{
y\in L^1(\Omega,\mathcal F,\PP):
y_L\leq y\leq y_U
\quad\PP\text{-a.s.}
\right\}.
\]

Absent additional information, every $y\in\Sel^1(Y)$ is a candidate for the latent outcome $y^*$. For a given selection $y$, the population best linear
predictor coefficient is defined by the normal equations
\[
\E\!\left[
\tilde x\bigl(y-\tilde x'\theta\bigr)
\right]
=
\mathbf 0.
\]
Because $Q$ is nonsingular, the coefficient generated by $y$ is uniquely
given by
\[
\theta(y)
=
Q^{-1}\E[\tilde x y].
\]

The unconstrained sharp identification region is therefore
\begin{equation}\label{eq:condexp_base_id}
\Theta^I
=
\left\{
Q^{-1}\E[\tilde x y]:
y\in\Sel^1(Y)
\right\}.
\end{equation}
Equivalently,
\[
\Theta^I
=
\left\{
\theta\in\R^{d+1}:
\E\!\left[
\tilde x\bigl(y-\tilde x'\theta\bigr)
\right]
=
\mathbf 0
\text{ for some }y\in\Sel^1(Y)
\right\}.
\]

It is useful to define the attainable cross-moment set
\[
\mathcal M
=
\left\{
\E[\tilde x y]:
y\in\Sel^1(Y)
\right\}.
\]
Then
\[
\Theta^I=Q^{-1}\mathcal M.
\]

\begin{proposition}[Basic properties of the unconstrained set]
\label{prop:basic}
Under Assumption \ref{as:main}, the cross-moment set $\mathcal M$ and the
identification region $\Theta^I=Q^{-1}\mathcal M$ are nonempty, compact, and
convex.
\end{proposition}

The proof represents $\mathcal M$ as the Aumann integral of an integrably
bounded, measurable, compact-valued, and convex-valued correspondence. The
relevant background results are collected in \cref{app:background} and all proofs for this section are in \cref{app:proofs}.

\subsection{Information on \texorpdfstring{$\E[y^*]$}{E[y*]}}\label{sec:knownExpectation}

We first consider the environment in which the researcher knows the
unconditional population expectation of the latent outcome.

\begin{assumption}[Known expected value]\label{assm:fixed_mean}
The unconditional expectation of the latent outcome is known:
\[
\E[y^*]=\kappa,
\]
where
\[
\kappa\in[\E(y_L),\E(y_U)].
\]
\end{assumption}

Under Assumption \ref{assm:fixed_mean}, candidate outcomes are restricted to
\[
\Sel^1(Y\mid\kappa)
=
\left\{
y\in\Sel^1(Y):
\E[y]=\kappa
\right\}.
\]

The compatibility condition
\[
\E[y_L]\leq\kappa\leq\E[y_U]
\]
guarantees that $\Sel^1(Y\mid\kappa)$ is nonempty. To see this, assume that
$\E[y_U-y_L]>0$, and define
\[
\lambda_\kappa
=
\frac{\kappa-\E[y_L]}{\E[y_U-y_L]}
\in[0,1].
\]
Then
\[
y_\kappa
=
y_L+\lambda_\kappa(y_U-y_L)
\]
belongs to $\Sel^1(Y\mid\kappa)$ and $\E[y_{\kappa}]=\kappa$. If $\E[y_U-y_L]=0$, then $P(y_L<y_U)=0$ implies that
$y_L=y_U$ almost surely, and Assumption \ref{assm:fixed_mean} implies
$\kappa=\E[y_L]=\E[y_U]$.

Given Assumption \ref{assm:fixed_mean}, we use the constrained selection set $\Sel^1(Y|\kappa)$ to define the identification region as
\begin{equation}\label{eq:condexp_mean_id}
\Theta^I_\kappa
=
\left\{
Q^{-1}\E[\tilde x y]:
y\in\Sel^1(Y\mid\kappa)
\right\}.
\end{equation}
The identification region in equation (\ref{eq:condexp_mean_id}) is not easy to compute directly by going over all selections in $\Sel^1(Y|\kappa)$. The following Proposition shows how to characterize the constrained identification set as an intersection of two sets that can be computed.
\begin{proposition}[Hyperplane characterization]
\label{prop:mean_hyperplane}
Under Assumptions \ref{as:main} and \ref{assm:fixed_mean},
\begin{equation}\label{eq:mean_hyperplane_intersection}
\Theta^I_\kappa
=
\Theta^I\cap\Theta_\kappa,
\end{equation}
where
\[
\Theta_\kappa
=
\left\{
\theta\in\R^{d+1}:
\E[\tilde x]'\theta=\kappa
\right\}.
\]
Consequently, $\Theta^I_\kappa$ is a nonempty compact convex set whose affine
dimension is at most $d$.

If $\Theta^I$ is a full-dimensional subset of $\R^{d+1}$ and
\[
\operatorname{ri}(\Theta^I)\cap\Theta_\kappa\neq\varnothing,
\]
then $\Theta^I_\kappa$ has affine dimension exactly $d$.\footnote{For a set $A \subset \R^d$, $ri(A)$ is the relative interior of $A$.}
\end{proposition}

\begin{remark}
The equality of the affine dimension with $d$ requires a nondegeneracy condition. Without it, the hyperplane may intersect $\Theta^I$ in a lower-dimensional face or at a single point.
\end{remark}

\subsection{Directional and Coordinate Bounds}\label{subsec:directional_bounds}

The set $\Theta^I_\kappa$ is obtained by intersecting $\Theta^I$ with the affine hyperplane, $\Theta_{\kappa}$, induced by the known mean. Thus, knowing the mean of $y^*$ can only reduce the identified set, and under nondegeneracy it lowers the affine dimension by one. Before we can measure how much the known mean tightens the identification region, we first characterize the boundary of the unconstrained set $\Theta^I$ itself. Researchers are often interested in the projection of an identification
region onto a specific linear combination of coordinates: given $\Theta^I$, what is the maximal or minimal value that $\theta_j$ can admit, for $j=0,1,\dots,d$? More generally, for $r\in\R^{d+1}$ representing a linear
combination of the parameters in $\theta$, we would like to find the maximal or minimal value that $r'\theta$ can admit over $\Theta^I$. This
apparatus is reused in \Cref{subsec:value_of_mean} below to quantify the
impact of the auxiliary information $\E[y^*]=\kappa$.

To characterize the boundary of the identified set, fix a direction
$r\in\R^{d+1}$ and define
\[
s_r
=
r'Q^{-1}\tilde x.
\]
For a coordinate direction $r=e_j$, write
\[
s_j=e_j'Q^{-1}\tilde x.
\]

For every selection $y\in\Sel^1(Y)$,
\[
r'\theta(y)
=
r'Q^{-1}\E[\tilde x y]
=
\E[s_r y].
\]

Let
\[
\Delta:=y_U-y_L\geq0.
\]
Every $y\in\Sel^1(Y)$ can be represented as
\[
y=y_L+\tau\Delta,
\]
where $\tau:\Omega\to[0,1]$ is measurable. 

\subsubsection*{Unconstrained bounds}

In the unconstrained problem, where $\E[y^*]$ is not given,
\[
\sup_{y\in\Sel^1(Y)}\E[s_r y]
=
\E[s_r y_L]
+
\sup_{0\leq\tau\leq1}\E[s_r\tau\Delta].
\]
Because there is no aggregate restriction on $\tau$,
an unconstrained maximizing allocation of $\tau$ satisfies $\tau^{r}(\omega)=1
\quad\text{on }\{\omega : s_r>0\}$, $
\tau^{r}(\omega)=0
\quad\text{on }\{\omega : s_r<0\}$, and with arbitrary values on $\{\omega : s_r=0\}$.

Therefore,
\begin{equation}\label{eq:unconstrained_upper}
\sup_{\theta\in\Theta^I}r'\theta
=
\E[s_r y_L]+\E[(s_r)^+\Delta].
\end{equation}

Similarly,
\begin{equation}\label{eq:unconstrained_lower}
\inf_{\theta\in\Theta^I}r'\theta
=
\E[s_r y_L]-\E[(s_r)^-\Delta].
\end{equation}

For bounds on $\theta_j$ we set $r=e_j$. The maximal and minimal values are
\[
\theta_j^{\max}
=
\E[s_jy_L]+\E[s_j^+\Delta]
\]
and
\[
\theta_j^{\min}
=
\E[s_jy_L]-\E[s_j^-\Delta].
\]

Equations (\ref{eq:unconstrained_upper}) and (\ref{eq:unconstrained_lower}) require the calculation of two expectations that include stochastic weights $s_r$. We next give a useful representation of the resulting directional and coordinate breadths.
\begin{definition}
The breadth of a nonempty compact set $C$ in direction $r$ is
\[
w_C(r)
=
\sup_{\theta\in C}r'\theta
-
\inf_{\theta\in C}r'\theta.
\]

\end{definition}

\begin{proposition}[Partial-regression representation of coordinate breadth]
\label{prop:width}
Under Assumption \ref{as:main}, for every $r\in\R^{d+1}$,
\begin{equation}\label{eq:dir_width}
w_{\Theta^I}(r)
=
\E[|s_r|\Delta].
\end{equation}

Moreover, for $j\in\{1,\dots,d\}$ let $\tilde x_j$ be a nonconstant coordinate and let $\tilde x_{-j}$ denote all remaining components of $\tilde x$. Since $Q$ is positive definite, its principal block $\E[\tilde x_{-j}\tilde x_{-j}']$ is also positive definite, hence nonsingular. Let $\tilde x_j^*$ be the residual from
the population linear projection of $\tilde x_j$ on $\tilde x_{-j}$:
\[
\tilde x_j^*
=
\tilde x_j
-
\E[\tilde x_j\tilde x_{-j}']
\E[\tilde x_{-j}\tilde x_{-j}']^{-1}
\tilde x_{-j}.
\]
Let
\[
\sigma_j^2
=
\E[(\tilde x_j^*)^2]>0.
\]
Then
\begin{equation}\label{eq:coord_score}
s_j
=
e_j'Q^{-1}\tilde x
=
\frac{\tilde x_j^*}{\sigma_j^2}
\qquad
\PP\text{-a.s.},
\end{equation}
and
\begin{equation}\label{eq:coord_width}
w_{\Theta^I}(e_j)=\theta_j^{\max}-\theta_j^{\min}
=
\frac{\E[|\tilde x_j^*|\Delta]}
{\E[(\tilde x_j^*)^2]}.
\end{equation}
\end{proposition}

The proof for Proposition \ref{prop:width} is in the Appendix. The coordinate formula in Proposition \ref{prop:width} is a Frisch--Waugh--Lovell type representation. In a point-identified linear projection, the coefficient on $\tilde x_j$ can be computed using the residualized regressor $\tilde x_j^*$. Here the same residualization determines the width of the identified interval for $\theta_j$. The only additional object is the interval width $\Delta=y_U-y_L$, which measures how much freedom the analyst has in choosing a selection from the observed interval. Expression \eqref{eq:coord_width} shows that partial identification of $\theta_j$ is not determined by the average interval width alone. It depends on the joint distribution of $\Delta$ and the magnitude of the partialled-out regressor $|\tilde x_j^*|$. The quantity in \eqref{eq:coord_width} can be consistently estimated using sample analogs of the expectations.

\subsection{The Value of a Known Mean}\label{subsec:value_of_mean}

We now use the apparatus of \Cref{subsec:directional_bounds} to quantify the
impact of knowing $\E[y^*]=\kappa$ on the identification region, by comparing
the directional bounds of $\Theta^I_\kappa$ with those of $\Theta^I$.

\subsubsection*{Mean-constrained bounds}

Under Assumption \ref{assm:fixed_mean}, the allocation $\tau$ must satisfy
\[
\E[y_L+\tau\Delta]=\kappa,
\]
or equivalently
\[
\E[\tau\Delta]
=
\kappa-\E[y_L]
=:
\alpha.
\]
Since $\kappa \in \left[\E[y_L],\E[y_U]\right]$,
$0\leq\alpha\leq\E[\Delta]$.

The sharp directional upper bound is therefore obtained from
\begin{equation}\label{eq:mean_knapsack}
\sup_{0\leq\tau\leq1}
\E[s_r\tau\Delta]
\qquad
\text{subject to}
\qquad
\E[\tau\Delta]=\alpha.
\end{equation}

The problem in (\ref{eq:mean_knapsack}) has a useful allocation interpretation. The variable $\tau(\omega)$ determines what fraction of the available interval width $\Delta(\omega)$ is assigned to the upper endpoint in state $\omega$. The mean restriction fixes the total amount of width that must be allocated:
$$
\E[\tau\Delta]=\alpha.
$$
The objective assigns value $s_r(\omega)$ to one unit of allocated width in state $\omega$. Thus, states with larger $s_r$ are more valuable for maximizing the directional bound $r'\theta$.

It is therefore natural to measure the size of a state not by its probability alone, but by how much interval width it contributes. Define the finite nonnegative measure
$$
\mu(A)=\E[\Delta 1_A], \qquad A\in\mathcal F.
$$
Under $\mu$, a set of states receives a weight equal to the expected interval width available on that set. With this notation, the constraint $\E[\tau\Delta]=\alpha$ becomes
$$
\int_\Omega \tau\,d\mu=\alpha,
$$
and the objective becomes
$$
\E[s_r\tau\Delta]=\int_\Omega s_r\tau\,d\mu.
$$
Hence \eqref{eq:mean_knapsack} can be written as
$$
\sup_{\tau}\int_\Omega s_r\tau\,d\mu
\quad\text{subject to}\quad
0\leq \tau\leq 1,\qquad
\int_\Omega \tau\,d\mu=\alpha.
$$
This is a continuous fractional-knapsack problem: allocate exactly $\alpha$ units of width-weighted mass to the states with the largest values of $s_r$, allowing fractional allocation at the cutoff if necessary.

\begin{proposition}[Directional contraction from a known mean]
\label{prop:directional_contraction_mean}
Let \Cref{as:main} and \Cref{assm:fixed_mean} hold, and suppose $\Lambda=E[\Delta]>0$. Fix $r\in\R^{d+1}$ and let
$$
s_r=r'Q^{-1}\tilde{x}, \qquad
p_\kappa=\frac{\kappa-E[y_L]}{E[\Delta]}, \qquad
u_r=\bar\mu(s_r>0),
$$
where $\bar\mu(A)=E[\Delta 1_A]/E[\Delta]$. Let $F_r$ denote the distribution function of $s_r$ under the
width-weighted measure $\bar\mu$:
$$
F_r(z)=\bar\mu(s_r\leq z).
$$
Let $q_r$ denote the corresponding generalized quantile function:
$$
q_r(t)=\inf\{z\in\mathbb R:F_r(z)\geq t\}, \qquad t\in(0,1).
$$
The values of $q_r$ at $0$ and $1$ are immaterial for the integrals below.

The constrained upper support value is
$$
\sup_{\theta\in\Theta^I_\kappa} r'\theta
=
E[s_r y_L]
+
\Lambda\int_{1-p_\kappa}^{1} q_r(t)\,dt.
$$
The unconstrained upper support value is
$$
\sup_{\theta\in\Theta^I} r'\theta
=
E[s_r y_L]
+
\Lambda\int_{1-u_r}^{1} q_r(t)\,dt.
$$

Define the loss in the upper support value from imposing the known mean by
$$
\delta_\kappa(r)
=
\sup_{\theta\in\Theta^I} r'\theta
-
\sup_{\theta\in\Theta^I_\kappa} r'\theta .
$$
Then
\begin{equation}\label{eq:mean_band_integral}
\delta_\kappa(r)
=
\Lambda
\int_{\min\{1-p_\kappa,\,1-u_r\}}^{\max\{1-p_\kappa,\,1-u_r\}}
|q_r(t)|\,dt
\geq 0.
\end{equation}
Consequently, the breadth of the identified set contracts in direction $r$ by
\begin{equation}\label{eq:breadth_contraction_mean}
w_{\Theta^I}(r)-w_{\Theta^I_\kappa}(r)
=
\delta_\kappa(r)+\delta_\kappa(-r).
\end{equation}
\end{proposition}

The proposition expresses the tightening from the known mean as an area under the width-weighted quantile function of the score $s_r$. The unrestricted optimizer corresponds to the quantile level $1-u_r$, i.e. the zero cutoff $s_r=0$. The point $1-p_\kappa$ is the cutoff required by the known mean. The contraction is the value lost when moving from the unconstrained cutoff to the mean-constrained cutoff. Hence the restriction has no first-order effect when the two cutoffs are close and the distribution of $s_r$ is smooth around zero.

The formula has a direct plug-in analogue. In a sample, one computes
$\hat s_{ri}=r'\hat Q^{-1}\tilde x_i$, weights observations by
$\Delta_i/\sum_j\Delta_j$, sorts the scores, and evaluates the weighted
area under the empirical quantile function between
$1-\hat p_\kappa$ and $1-\hat u_r$. The full sample formula is given in
Appendix \ref{app:numerical}.

We now show that small deviations in the width budget have only second order effect on the directional support value.
\begin{corollary}[Local cost of the mean restriction]
\label{cor:local_mean_contraction}
Suppose that the distribution of $s_r$ under $\bar\mu$ has a density $f_r$
in a neighborhood of zero, and that $f_r$ is continuous with $f_r(0)>0$.
Then, as $p_\kappa\to u_r$,
$$
\delta_\kappa(r)
=
\frac{\Lambda}{2f_r(0)}
(p_\kappa-u_r)^2
+
o\left((p_\kappa-u_r)^2\right).
$$
\end{corollary}

\subsection{Identification with Transformations}\label{subsec:transformations}

Up to this point we have studied the best linear predictor of the latent outcome $y^*$. In many applications, however, the parameter of interest is instead the best linear predictor of a transformation $f(y)$, or a known population moment of $f(y^*)$ is used as auxiliary information about $y^*$ itself. Examples include logarithms of income, indicator functions, or nonlinear utility transformations. We consider both uses of a transformation $f$ in this subsection, since they share the same regularity conditions and rely on the same splicing argument, but lead to identification regions with markedly different geometry.

Let
\[
f:\R\to\R
\]
be Borel measurable.

\begin{assumption}[Uniform integrability of the transformation] \label{as:int_f}
There exists a nonnegative measurable random variable $F$ such that
$|f(y)|\leq F$ for every $y\in \Sel^1(Y)$ $P$-a.s., and
$$
E[F^2]<\infty.
$$
\end{assumption}
Since Assumption \ref{as:main} already gives $E[|\tilde x|^2]<\infty$, Assumption \ref{as:int_f} also implies $E[|\tilde x|F]<\infty$ by Cauchy-Schwarz.

\begin{assumption}[Atomlessness]\label{as:atomless}
The probability space $(\Omega,\mathcal F,\PP)$ is atomless.
\end{assumption}

\subsubsection*{Retargeting to the transformed outcome}

We first retarget the object of interest to the best linear predictor of $f(y^*)$ itself,
\[
\theta_f=Q^{-1}\E[\tilde x f(y^*)],
\]
maintaining the restriction
\[
\E[y^*]=\kappa.
\]
The key difference from the previous section is that the objective now depends on $f(y^*)$, while the auxiliary information continues to constrain the untransformed variable through $E[y^*]=\kappa$. Unless $f$ is affine, the objective is nonlinear in $\tau$, so the fractional-knapsack representation of Section \ref{subsec:value_of_mean} no longer applies. Proposition \ref{prop:transformed_convexity} below shows that convexity nevertheless survives under nonatomlessness. 

Define,
\begin{equation}\label{eq:transformed_moment_set}
\mathcal M_f(\kappa)
=
\left\{
\E[\tilde x f(y)]:
y\in\Sel^1(Y\mid\kappa)
\right\}
\end{equation}
and
\begin{equation}\label{eq:transformed_theta_set}
\Theta^I_{f\mid\kappa}
=
Q^{-1}\mathcal M_f(\kappa).
\end{equation}

Although the restriction on admissible selections is unchanged, the optimization problem is generally no longer linear because the objective depends nonlinearily on the allocation variable through
\[
f(y_L+\tau\Delta).
\]
Unless $f$ is affine, the optimization is no longer linear in $\tau$. Consequently, the fractional-knapsack representation of  \Cref{subsec:value_of_mean}. The following Proposition shows that despite this loss of linearity, convexity of the identified region is preserved.

\begin{proposition}[Convexity under arbitrary Borel transformations]
\label{prop:transformed_convexity}
Under Assumptions \ref{as:main}, \ref{assm:fixed_mean},
\ref{as:int_f}, and \ref{as:atomless}, the sets
\[
\mathcal M_f(\kappa)
\quad\text{and}\quad
\Theta^I_{f\mid\kappa}
\]
are convex.
\end{proposition}

The proposition shows that the geometric properties established earlier survive the introduction of arbitrary measurable transformations. Although the optimization problem becomes nonlinear, the attainable moment set remains convex because convexity is generated by the atomless probability space through Lyapunov's theorem rather than by linearity of $f$.

For every $r\in\R^{d+1}$, let
\begin{equation}\label{eq:transformed_support}
h_{\Theta^I_{f\mid\kappa}}(r)
=
\sup_{\theta\in\Theta^I_{f\mid\kappa}}r'\theta
=
\sup_{y\in\Sel^1(Y\mid\kappa)}
\E[s_r f(y)].
\end{equation}

Convexity of the transformation itself yields an additional implication. Since every admissible selection satisfies $E[y]=\kappa$, Jensen's inequality implies
\begin{equation}\label{eq:jensen_restriction}
\E[\tilde x]'\theta_f
=
\E[f(y)]
\geq
f(\E[y])
=
f(\kappa).
\end{equation}
If $f$ is concave, the inequality is reversed. Thus, even though the restriction is imposed only on the mean of the latent outcome, convexity (concavity) of $f$ automatically generates a lower (upper) bound on the mean of every admissible transformed outcome.

The preceding discussion assumes that only $E[y^*]$ is known. If the researcher additionally knows the population mean of the transformed outcome itself,
\[
\E[f(y^*)]=\kappa_f,
\]
then the transformed problem reduces to the same hyperplane characterization developed in \Cref{sec:knownExpectation}, namely
\[
\E[\tilde x]'\theta_f=\kappa_f.
\]
Thus, auxiliary information about the transformed outcome can be incorporated exactly as before.
\subsubsection*{Auxiliary moment restrictions without retargeting}

We now consider the complementary case: the target remains the original
$\theta=Q^{-1}\E[\tilde x y^*]$, and a known moment of the transformation,
$\E[f(y^*)]=\kappa_f$, is used purely as auxiliary information narrowing
the selections admissible for $y^*$. No restriction on $\E[y^*]$ itself
is imposed.

Define the achievable moment set
\[
\mathcal K_f=\{\E[f(y)]:y\in\Sel^1(Y)\},
\]
and, for $\kappa_f\in\mathcal K_f$,
\[
\Sel^1(Y\mid f,\kappa_f)
=
\left\{
y\in\Sel^1(Y):\E[f(y)]=\kappa_f
\right\}.
\]
The compatibility condition, $\kappa_f\in\mathcal K_f$, guarantees this set is nonempty; unlike the linear case of Assumption~\ref{assm:fixed_mean}, $\mathcal K_f$ need not be an interval with known endpoints in closed form, so compatibility is stated directly as membership in
$\mathcal K_f$ rather than via an explicit interval.

The resulting sharp identification region for the \emph{original}
target is
\begin{equation}\label{eq:moment_aux_set}
\Theta^I_{\kappa_f}
=
\left\{
Q^{-1}\E[\tilde x y]:
y\in\Sel^1(Y\mid f,\kappa_f)
\right\}.
\end{equation}

\begin{proposition}[Convexity under an auxiliary moment restriction]
\label{prop:moment_aux_convexity}
Under Assumptions~\ref{as:main}, \ref{as:int_f}, and
\ref{as:atomless}, $\mathcal K_f$ is convex, and for every
$\kappa_f\in\mathcal K_f$, $\Theta^I_{\kappa_f}$ is nonempty and convex.
\end{proposition}

The proof applies Lyapunov's theorem for vector measures to the pair $(f(y_1)-f(y_0),\,\tilde x(y_1-y_0))$ for any two selections $y_0,y_1\in\Sel^1(Y\mid f,\kappa_f)$: since both satisfy $\E[f(y_0)]=\E[f(y_1)]=\kappa_f$, the first component of this pair has mean zero, so splicing $y_0$ and $y_1$ along the sets furnished by Lyapunov's theorem preserves the moment restriction exactly while tracing out the line segment between $Q^{-1}\E[\tilde x y_0]$ and $Q^{-1}\E[\tilde x y_1]$ in $\theta$-space; see \cref{app:background} for the background result and  \cref{app:proofs} for the full argument.

The geometric distinction is important. Unlike Proposition~\ref{prop:mean_hyperplane}, $\Theta^I_{\kappa_f}$ need not collapse to a lower-dimensional set. The restriction $\E[y^*]=\kappa$ is linear in $\theta$ because $\tilde x$'s first
coordinate is $1$, so $\E[\tilde x]'\theta=\E[y]$ directly, pinning $\theta$ to an affine hyperplane. A nonlinear moment restriction $\E[f(y^*)]=\kappa_f$ has no such direct algebraic counterpart in $\theta$-space: it removes some selections from $\Sel^1(Y)$ without
confining the resulting $\theta$ values to any fixed hyperplane, so $\Theta^I_{\kappa_f}$ is generically full-dimensional -- a narrowed region rather than a segment. This is the same regularity condition and the same splicing argument as Proposition~\ref{prop:transformed_convexity}, applied to a different target; the two propositions differ only in which functional of $y$ the objective retains and which the constraint restricts, and it is exactly this difference that separates a lower-dimensional segment from a full-dimensional band.

\begin{remark}[Quantile information]\label{rmk:quantile-information}
Quantile information is a special case of a transformation restriction, with $f(y)=\1\{y\leq m\}$ or a closely related one-sided indicator. Because such indicator transformations are discontinuous, their support problem requires separate treatment.  We therefore state the general transformation and quantile extensions, including the corresponding Lagrangian representation, in \cref{app:extensions}.
\end{remark}
\subsection{Information on \texorpdfstring{$\E[y^*\mid x_1]$}{E[y*|x1]}}\label{subsec:conditional_x1}
We now consider a stronger form of auxiliary information in which the researcher knows the conditional mean of the latent outcome given a subset of the covariates, and we show how this restriction reduces the full identification problem to the coefficient block associated with the remaining covariates.

Partition
\[
x=(x_1',x_2')',
\]
where
\[
x_1\in\R^{d_1},
\qquad
x_2\in\R^{d_2},
\qquad
d_1+d_2=d.
\]
Define
\[
w
=
\begin{pmatrix}
1\\x_1
\end{pmatrix}.
\]
Then
\[
\tilde x
=
\begin{pmatrix}
w\\x_2
\end{pmatrix}.
\]

\begin{assumption}[Known conditional sub-moment]
\label{assm:conditional_x1}
The conditional mean of the latent outcome given $x_1$ is known:
\[
\E[y^*\mid x_1]=\kappa(x_1)
\qquad
\PP\text{-a.s.},
\]
where
\[
\E[y_L\mid x_1]
\leq
\kappa(x_1)
\leq
\E[y_U\mid x_1]
\qquad
\PP\text{-a.s.}
\]
\end{assumption}

Define
\[
\Sel^1(Y\mid\kappa(x_1))
=
\left\{
y\in\Sel^1(Y):
\E[y\mid x_1]=\kappa(x_1)
\quad\PP\text{-a.s.}
\right\}.
\]

This set is nonempty. Let
\[
D(x_1)=\E[\Delta\mid x_1]
\]
and define
\[
\lambda(x_1)
=
\begin{cases}
\dfrac{
\kappa(x_1)-\E[y_L\mid x_1]
}{
D(x_1)
},
&
D(x_1)>0,
\\[3mm]
0,
&
D(x_1)=0.
\end{cases}
\]
On the event $D(x_1)=0$, nonnegativity of $\Delta$ implies $\Delta=0$ conditionally almost surely, and compatability therefore gives $\kappa(x_1)=E[y_L|x_1]=E[y_U|x_1]$. Then from the definition above, $0\leq\lambda(x_1)\leq1$ almost surely and
\[
y_\kappa
=
y_L+\lambda(x_1)\Delta
\]
satisfies
\[
\E[y_\kappa\mid x_1]=\kappa(x_1).
\]

The sharp identified set is
\begin{equation}\label{eq:conditional_x1_sharp}
\Theta^I_{\kappa(x_1)}
=
\left\{
Q^{-1}\E[\tilde x y]:
y\in\Sel^1(Y\mid\kappa(x_1))
\right\}.
\end{equation}

Partition
\[
\theta
=
\begin{pmatrix}
\theta_1\\
\theta_2
\end{pmatrix},
\qquad
\theta_1\in\R^{d_1+1},
\quad
\theta_2\in\R^{d_2},
\]
and define
\[
\Sigma_{11}=\E[ww'],
\qquad
\Sigma_{12}=\E[wx_2'],
\]
\[
\Sigma_{21}=\Sigma_{12}',
\qquad
\Sigma_{22}=\E[x_2x_2'].
\]
Then
\[
Q
=
\begin{pmatrix}
\Sigma_{11}&\Sigma_{12}\\
\Sigma_{21}&\Sigma_{22}
\end{pmatrix}.
\]

Define
\[
g_1
=
\E[w\kappa(x_1)]
\]
and the linear-projection residual
\[
x_2^*
=
x_2-\Sigma_{21}\Sigma_{11}^{-1}w.
\]
Then
\[
\E[x_2^*w']=0
\]
and
\[
\Sigma_{2\cdot1}
=
\E[x_2^*x_2^{*\prime}]
=
\Sigma_{22}
-
\Sigma_{21}\Sigma_{11}^{-1}\Sigma_{12}
\]
is positive definite.

Define
\begin{equation}\label{eq:theta2_set}
\Theta_2^I
=
\left\{
\Sigma_{2\cdot1}^{-1}\E[x_2^*y]:
y\in\Sel^1(Y\mid\kappa(x_1))
\right\}.
\end{equation}

\begin{proposition}[Sharp FWL representation]
\label{prop:conditional_fwl}
Under Assumptions \ref{as:main} and \ref{assm:conditional_x1},
\begin{equation}\label{eq:fwl_lifted_set}
\Theta^I_{\kappa(x_1)}
=
\left\{
(\theta_1',\theta_2')':
\begin{array}{l}
\theta_2\in\Theta_2^I,\\[1mm]
\theta_1
=
\Sigma_{11}^{-1}
(g_1-\Sigma_{12}\theta_2)
\end{array}
\right\}.
\end{equation}
This representation holds whether $x_1$ has finite, countable, or continuous
support.
\end{proposition}

Consequently, $\Theta^I_{\kappa(x_1)}$ is the image of $\Theta_2^I$ under an injective affine map and therefore has affine dimension at most $d_2$. In particular, the conditional mean restriction imposes $d_1 +1$ exact linear restrictions on the full coefficient vector.

If $x_1$ has finite support, the conditional restriction is a finite collection of cell-specific moment restrictions. If $x_1$ is continuously distributed, it is an infinite-dimensional conditional restriction. The block normal-equation decomposition is the same in both cases.

An alternative residual is
\[
\bar x_2=x_2-\E[x_2\mid x_1].
\]
It satisfies
\[
\E[\bar x_2\mid x_1]=0
\]
and therefore
\[
\E[\bar x_2y]
=
\E[\bar x_2\{y-\kappa(x_1)\}]
\]
for every conditionally admissible selection, provided that these expectations exist. This conditional-mean
residualization corresponds to partialling out an unrestricted function of $x_1$ and should be distinguished from the finite-dimensional linear residual $x_2^*$.

\subsection{Information on \texorpdfstring{$\E[y^*\mid v]$}{E[y*|v]}}\label{subsec:conditional_v}

We next consider auxiliary conditional-mean information indexed by an external variable v that does not enter the best linear predictor. Unlike conditioning on a subvector of the regressors, this restriction does not generally impose a direct linear restriction on the coefficient vector. Instead, it sharpens identification by fixing how the total interval-width allocation must be distributed across values of $v$.

Let $v$ be an observed external variable that need not enter the BLP
specification.

\begin{assumption}[External conditional moment]
\label{assm:external}
The conditional mean of the latent outcome given $v$ is known:
\[
\E[y^*\mid v]=\kappa(v)
\qquad
\PP\text{-a.s.},
\]
where
\[
\E[y_L\mid v]
\leq
\kappa(v)
\leq
\E[y_U\mid v]
\qquad
\PP\text{-a.s.}
\]
\end{assumption}

Define
\[
\Sel^1(Y\mid\kappa(v))
=
\left\{
y\in\Sel^1(Y):
\E[y\mid v]=\kappa(v)
\quad\PP\text{-a.s.}
\right\}.
\]

As before, conditional compatibility implies nonemptiness by taking
\[
y_v
=
y_L+\lambda(v)\Delta,
\]
where
\[
\lambda(v)
=
\begin{cases}
\dfrac{
\kappa(v)-\E[y_L\mid v]
}{
\E[\Delta\mid v]
},
&
\E[\Delta\mid v]>0,
\\[3mm]
0,
&
\E[\Delta\mid v]=0.
\end{cases}
\]

On the event $E[\Delta|v]=0$, nonnegativity of $\Delta$ implies that $\Delta=0$ conditionally almost surely. Compatibility, therefore, gives $\kappa(v)=E[y_L|v]=E[y_U|v]$ a common conditional expectation for all selections. As a result, $E[y_v|v]=\kappa(v)$

The sharp identified set is
\begin{equation}\label{eq:external_sharp_set}
\Theta^I_{\kappa(v)}
=
\left\{
Q^{-1}\E[\tilde x y]:
y\in\Sel^1(Y\mid\kappa(v))
\right\}.
\end{equation}

\begin{proposition}[Sharpness under an external conditional mean]
\label{prop:external_sharp}
Under Assumptions \ref{as:main} and \ref{assm:external},
$\Theta^I_{\kappa(v)}$ is the sharp identification region and
\[
\Theta^I_{\kappa(v)}\subseteq\Theta^I_{\kappa}\subseteq\Theta^I.
\]
\end{proposition}

Unlike the restriction studied in \Cref{subsec:conditional_x1}, the external conditional mean generally has no direct finite-dimensional representation in coefficient space because v is excluded from $\Tilde{x}$. Its identifying content is instead revealed through the support-function comparison below.

Fix $r\in\R^{d+1}$. Then
\begin{equation}\label{eq:external_support}
h_{\Theta^I_{\kappa(v)}}(r)
=
\sup_{y\in\Sel^1(Y\mid\kappa(v))}
\E[s_ry].
\end{equation}

Writing $y=y_L+\tau\Delta$, the conditional restriction becomes
\[
\E[\tau\Delta\mid v]
=
\kappa(v)-\E[y_L\mid v]
=:
\alpha(v).
\]

\subsubsection*{Discrete external variable}

Suppose
\[
v\in\{v_1,\dots,v_M\},
\qquad
p_m=\PP(v=v_m)>0.
\]
Let $\tau_m$ denote the allocation rule within event $\{v=v_m\}$.  Then
\begin{align}
h_{\Theta^I_{\kappa(v)}}(r)
&=
\sum_{m=1}^M
p_m
\Bigg[
\E[s_ry_L\mid v=v_m]
\nonumber\\
&\qquad
+
\sup_{\substack{0\leq\tau_m\leq1\\
\E[\tau_m\Delta\mid v=v_m]=\alpha(v_m)}}
\E[s_r\tau_m\Delta\mid v=v_m]
\Bigg].
\label{eq:discrete_v_support}
\end{align}

Define the cell measures
\[
\mu_m(A)
=
\E[\Delta\1_{A\cap\{v=v_m\}}],
\]
and the cell budgets
\[
A_m
=
\E[
\{\kappa(v_m)-\E[y_L\mid v=v_m]\}
\1\{v=v_m\}
].
\]
Let
\[
\kappa:=\E[\kappa(v)]=\E[y^*],
\qquad
\alpha=\kappa-\E[y_L].
\]
Then
\[
\sum_{m=1}^M A_m=\alpha.
\]

For a finite measure $\nu$, define the upper-tail functional
\[
\mathcal T_a^\nu(Z)
=
\sup\left\{
\int\tau Z\,d\nu:
0\leq\tau\leq1,\
\int\tau\,d\nu=a
\right\},
\qquad
0\leq a\leq\nu(\Omega).
\]

Let
\[
\mu(A)=\E[\Delta\1_A].
\]
Then
\[
\mu(\Omega)=\E[\Delta]=\sum_{m=1}^M\mu_m(\Omega).
\]

\begin{proposition}[Additional contraction from conditioning]
\label{prop:external_contraction}
Under Assumptions \ref{as:main} and \ref{assm:external},
\begin{equation}\label{eq:external_split_conditional}
h_{\Theta^I_{\kappa(v)}}(r)
=
\E[s_ry_L]
+
\sum_{m=1}^M
\mathcal T_{A_m}^{\mu_m}(s_r),
\end{equation}
whereas
\begin{equation}\label{eq:external_split_pooled}
h_{\Theta^I_\kappa}(r)
=
\E[s_ry_L]
+
\max_{\substack{
0\leq a_m\leq\mu_m(\Omega)\\
\sum_{m=1}^M a_m=\alpha
}}
\sum_{m=1}^M
\mathcal T_{a_m}^{\mu_m}(s_r).
\end{equation}
Consequently,
\begin{align}
h_{\Theta^I_\kappa}(r)
-
h_{\Theta^I_{\kappa(v)}}(r)
&=
\max_{\substack{
0\leq a_m\leq\mu_m(\Omega)\\
\sum_m a_m=\alpha
}}
\sum_{m=1}^M\mathcal T_{a_m}^{\mu_m}(s_r)
-
\sum_{m=1}^M\mathcal T_{A_m}^{\mu_m}(s_r)
\nonumber\\
&\geq0.
\label{eq:external_extra}
\end{align}

Equality holds if and only if the imposed split $(A_1,\dots,A_M)$ is value-maximizing in \eqref{eq:external_split_pooled}. If $0<\alpha<\mu(\Omega)$, there is
a common cutoff $c^*$ and tie fractions $\gamma_m\in[0,1]$ such that
\[
A_m
=
\mu_m(s_r>c^*)
+
\gamma_m\mu_m(s_r=c^*)
\]
for every $m$. If $\alpha=0$ or $\alpha=\mu(\Omega)$, equality holds trivially because the allocation is uniquely fixed, up to states with zero interval width.

If $s_r$ is independent of $v$ under the normalized width measure $\Bar{\mu}$ and the imposed cell budgets are proportional to their width masses,
\[
A_m
=
\frac{\mu_m(\Omega)}{\mu(\Omega)}\alpha,
\]
then the additional contraction is zero.
\end{proposition}

\begin{corollary}[Local cutoff-dispersion approximation]
\label{cor:between_cell}
Consider a sequence of problems for which
\[
\max_{1\leq m\leq M}|c_m^*-c^*|\to0,
\]
where $c_m^*$ solves the $m$-th cell problem at budget $A_m$ and $c^*$ is the
pooled cutoff. Assume additionally that the densities are uniformly continuous and uniformly bounded in neighborhoods containing $c^*$ and $c^*_m$ along the sequence.

Suppose $0<A_m<\mu_m(\Omega)$ for every $m$. Suppose also that under the normalized measure
\[
\frac{\mu_m}{\mu_m(\Omega)},
\]
the score $s_r$ has a density $f_m$ that is continuous and strictly positive
in a neighborhood of $c^*$. Then
\begin{align}
h_{\Theta^I_\kappa}(r)
-
h_{\Theta^I_{\kappa(v)}}(r)
&=
\frac12
\sum_{m=1}^M
\mu_m(\Omega)f_m(c^*)
(c_m^*-c^*)^2
\nonumber\\
&\quad
+
o\!\left(
\sum_{m=1}^M
\mu_m(\Omega)(c_m^*-c^*)^2
\right).
\label{eq:between_cell_expansion}
\end{align}
\end{corollary}

Thus, locally, the additional contraction is governed by dispersion of the cell-specific optimal cutoffs around the pooled cutoff. Cells contribute more to the additional contraction when their imposed budgets force their optimal cutoffs farther from the pooled cutoff, with the local cost weighted by the amount of interval width in the cell and by the width-weighted density of the score at the pooled cutoff. This is a statement about cutoff dispersion, not a general identity with between-group variance
of the score.

\subsubsection*{Continuous external variable}

When v is continuously distributed, the same logic applies pointwise in $v$: the conditional restriction fixes a separate width budget at almost every value of $v$, and the support value is obtained by integrating the corresponding conditional knapsack values. 

Suppose $v$ takes values in a standard Borel space, a regular conditional distribution of $(y_L,y_U,x)$ given $v$ exists, and the conditional optimization admits a measurable selection of optimizers. Then
\begin{align}
h_{\Theta^I_{\kappa(v)}}(r)
&=
\E\Bigg[
\E[s_ry_L\mid v]
\nonumber\\
&\qquad
+
\sup_{\substack{0\leq\tau\leq1\\
\E[\tau\Delta\mid v]=\alpha(v)}}
\E[s_r\tau\Delta\mid v]
\Bigg].
\label{eq:continuous_v_support}
\end{align}

For almost every $t$, define
\[
K_r(t)
=
\sup_{\substack{0\leq\tau\leq1\\
\E[\tau\Delta\mid v=t]=\alpha(t)}}
\E[s_r\tau\Delta\mid v=t].
\]
If $t\mapsto K_r(t)$ is measurable, then
\[
h_{\Theta^I_{\kappa(v)}}(r)
=
\E\!\left[
\E[s_ry_L\mid v]+K_r(v)
\right].
\]
If $v$ has density $f_v$, this becomes
\[
h_{\Theta^I_{\kappa(v)}}(r)
=
\int
\left[
\E[s_ry_L\mid v=t]
+
K_r(t)
\right]
f_v(t)\,dt.
\]

\section{Numerical Illustration}\label{sec:numerical}

This section provides a numerical illustration of the identification results in \Cref{sec:condexp}. The exercise is not intended as a substantive analysis of the returns to education. Instead, we use Current Population Survey (CPS) Annual Social and Economic Supplement (ASEC) observed income to construct interval-valued outcomes and then compare the identified regions obtained under different forms of auxiliary information. Throughout this section, expectations and identified regions are implemented using their empirical analogs. To avoid excessive notation, we retain the population notation from \Cref{sec:condexp}. For simplicity, the numerical illustration treats the retained sample as an equally weighted empirical distribution. The exercise is intended to illustrate the geometry of the identification results rather than to provide population-representative estimates.

We restrict the sample to respondents with positive wage and salary income (\texttt{WSAL\_VAL}) who report being employed, rescale income to units of \$1{,}000, and drop the top percentile of the income distribution as a cosmetic trim against extreme outliers. The variable measuring years of completed education (\texttt{educ\_numeric}) enters the baseline covariate vector $\tilde x=(1,\mathrm{educ})$. Two further variables are retained for the extensions taken up later in this section: a categorical race indicator constructed from \texttt{PRDTRACE} (equal to $1$ for white respondents and $2$ for non-white), used in Section~\ref{subsec:conditional_x1_num} as $x_1$; and age (\texttt{A\_AGE}), used in  section~\ref{subsec:conditional_v_num} as the external variable $v$. After these restrictions the working sample contains $n=22{,}397$ observations, with mean income of \$63{,}990 and mean educational attainment of $14.2$ years.

Because $y^*$ (income) is observed as a point value in the underlying CPS extract, we generate the interval $Y=[y_L,y_U]$ used throughout this section using a stylized interval-privacy mechanism based on \citet{DingDing2022}, who introduce \emph{interval privacy} as a privacy criterion distinct from differential privacy: rather than perturbing a respondent's value with additive noise, the mechanism narrows it to a random range that provably contains the truth, constructed so that the range's conditional distribution given the true value is uninformative beyond the range itself. For each observation $i$, we draw two quantile indices independently of $y^*_i$, convert them into income anchors using the empirical income distribution, and order the resulting values. Together with the lower and upper endpoints of the empirical support, these anchors partition the income domain into three intervals. The released interval $[y_{L,i},y_{U,i}]$ is the unique partition cell containing $y^*_i$. Because the anchors are constructed without reference to the respondent's own value, the resulting interval is not centered on $y_i^*$, which is what makes it a nontrivial input for the identification results below. A construction that released an interval centered exactly at $y^*_i$ would reveal the latent outcome through the interval midpoint, defeating the purpose of treating the outcome as interval-valued at all.

\subsection{The unconstrained region}\label{subsec:unconstrained_region_num}

Figure~\ref{fig:unconstrained} plots $\Theta^I$ for
$\tilde x=(1,\mathrm{educ})$, computed via the closed-form directional bounds \eqref{eq:unconstrained_upper}--\eqref{eq:unconstrained_lower} implied by Proposition~\ref{prop:basic}. The region is an elongated and negatively sloped polygon, reflecting a tradeoff between the intercept and education coefficient induced by the joint distribution of education and the interval endpoints: selections that generate higher fitted levels at low education values tend to require lower education slopes, and conversely.

\begin{figure}[htbp]
\centering
\includegraphics[width=0.55\textwidth]{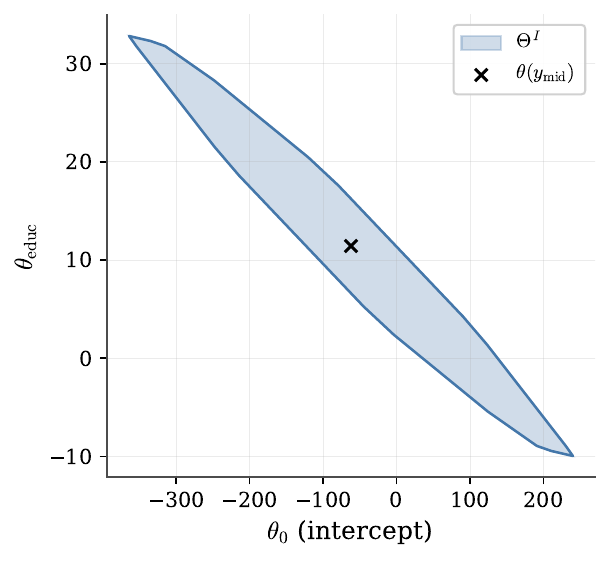}
\caption{Unconstrained sharp identification region $\Theta^I$.}
\label{fig:unconstrained}
\end{figure}

\subsection{The value of a known mean}

Imposing $\E[y^*]=\kappa$ at the empirical mean of the underlying uncoarsened and unweighted income collapses $\Theta^I_\kappa$ to a one-dimensional segment (Figure~\ref{fig:mean_known}), exactly as Proposition~\ref{prop:mean_hyperplane} predicts: $\Theta_\kappa$ is a single linear equation in a two-dimensional $\theta$, so its intersection with $\Theta^I$ has area exactly zero, not merely small. Because Proposition~\ref{prop:mean_hyperplane} already establishes that this intersection is exactly one-dimensional, we compute its endpoints directly from the support function evaluated at the single direction spanning $\Theta_\kappa$ and its reverse, rather than by sampling many directions and intersecting the resulting halfspaces as in Figure~\ref{fig:unconstrained}. This generic method presumes a two-dimensional interior point to anchor the construction, which a segment does not have, and can return a spuriously nonzero area from sampling noise alone.

\begin{figure}[htbp]
\centering
\includegraphics[width=0.55\textwidth]{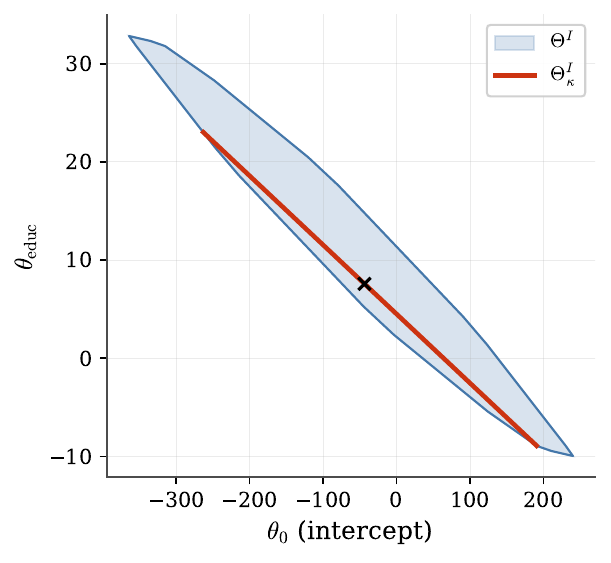}
\caption{Mean-constrained region $\Theta^I_\kappa$ inside $\Theta^I$.}
\label{fig:mean_known}
\end{figure}

\subsection{Identification with transformations}

We retarget the object of interest to $\theta_f=Q^{-1}\E[\tilde x f(y^*)]$, $f(y)=y^2$, and distinguish two cases according to which population moment is taken as known, since they have markedly different consequences for the geometry of the identified region.

If only $\E[y^*]=\kappa$ is known, Proposition~\ref{prop:transformed_convexity} guarantees that the corresponding population region is convex under atomlessness. The restriction does not impose a hyperplane on $\theta_f$ because there is no direct algebraic link between $\E[y^*]$ and $\E[\tilde x f(y^*)]$ once $f$ is nonlinear, so the known mean prunes admissible selections without confining $\theta_f$ to a lower-dimensional set. The resulting region remains a genuine two-dimensional band, narrower than the unconstrained $\Theta^I_f$. We do not plot the finite-sample analogue of this case separately. 

If instead $\E[f(y^*)]=\kappa_f$ is known, $\Theta^I_{f\mid\kappa_f}$ collapses to a one-dimensional segment (Figure~\ref{fig:transformed}), by exactly the hyperplane argument of Proposition~\ref{prop:mean_hyperplane}, now applied to $g=f(y)$ in place of $y$: since $\theta_f$ is itself built from $\E[\tilde x f(y)]$, knowing $\E[f(y^*)]=\kappa_f$ pins $\E[\tilde x]'\theta_f=\kappa_f$ directly. This is the informative case, and it holds regardless of whether $\E[y^*]=\kappa$ is imposed in addition.

\begin{figure}[htbp]
\centering
\includegraphics[width=0.55\textwidth]{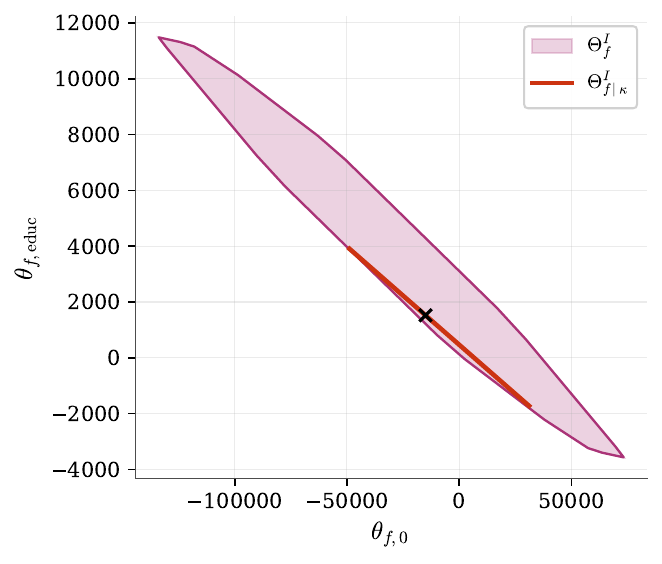}
\caption{Unconstrained $\Theta^I_f$ and $\kappa_f$-constrained $\Theta^I_{f\mid\kappa_f}$, $f(y)=y^2$.}
\label{fig:transformed}
\end{figure}

\subsection{An auxiliary moment restriction without retargeting}

Figure~\ref{fig:moment_only} imposes only $\E[y^{*2}]=\kappa_f$, with no restriction on $\E[y^*]$ and the target left at the original $\theta$.
Unlike the mean restriction, the $\E[y^{*2}]=\kappa_f$ restriction does not collapse the dimension of the identification region (see  Proposition~\ref{prop:moment_aux_convexity}).\footnote{Areas are computed by sampling the support function over a fine grid of directions and constructing the convex hull of the resulting supporting halfplanes. Proposition~\ref{prop:moment_aux_convexity} establishes convexity of the population object $\Theta^I_{\kappa_f}$ via Lyapunov's theorem, which requires an atomless probability space; the empirical distribution used here is finite and atomic, so this convexity is not automatically inherited by the finite-sample analog. The reported region should therefore be read as an outer bound on the finite-sample identified set obtained by directional optimization, rather than as a verified reconstruction of that set; each additional sampled direction can only tighten this bound toward the true region.}
The restricted region remains two-dimensional but retains only $16.4\%$ of the area of the unrestricted identification region. The restricted identification region is not a line segment since a nonlinear moment has no direct algebraic counterpart in $\theta$-space. This is the same mechanism as the mean-only case of \Cref{subsec:transformations} above, with the roles of the original and transformed outcome reversed.

\begin{figure}[htbp]
\centering
\includegraphics[width=0.55\textwidth]{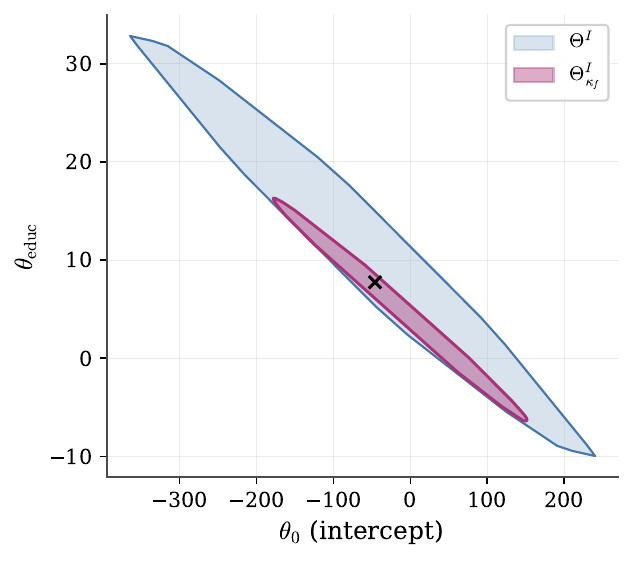}
\caption{$\Theta^I$ versus $\Theta^I_{\kappa_f}$, second moment known, mean not known.}
\label{fig:moment_only}
\end{figure}

\subsection{Information on \texorpdfstring{$\E[y^*\mid x_1]$}{E[y*|x1]}}\label{subsec:conditional_x1_num}

Extending the specification to $\tilde x=(1,\mathrm{race},\mathrm{educ})$ and imposing $\E[y^*\mid\mathrm{race}]=\kappa(\mathrm{race})$
(see Proposition~\ref{prop:conditional_fwl}) reduces the identified set to a one-dimensional affine segment parametrized by the return to education. For each admissible value of $\theta_{educ}$, Proposition \ref{prop:conditional_fwl} uniquely determines the corresponding intercept and race coefficient through the affine FWL relation.
The unconstrained interval for $\theta_{\mathrm{educ}}$,
$[-9.79,\ 32.89]$, narrows to $[-8.52,\ 22.90]$ once race is conditioned
on --- $73.6\%$ of the original width (Figure~\ref{fig:conditional_race}).
The full three-dimensional region, and the exact projection confirming this width, are reported in Appendix~\ref{app:numerical}. Note that the unconstrained bounds for $\theta_{educ}$ reported here differ slightly from those in Sections \ref{subsec:unconstrained_region_num} and \ref{subsec:conditional_v_num} because the specification here additionally includes race as a regressor.

\begin{figure}[htbp]
\centering
\includegraphics[width=0.6\textwidth]{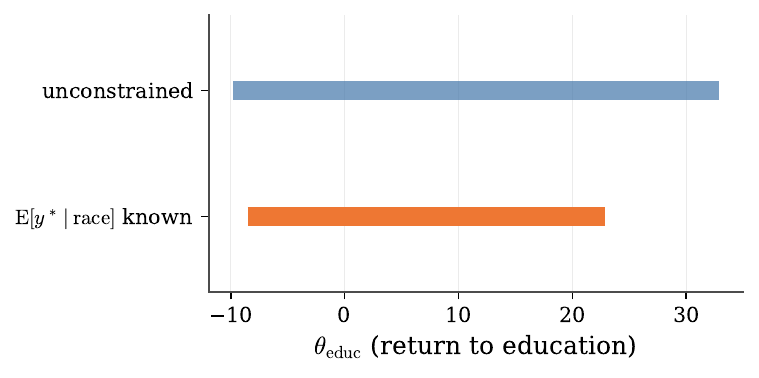}
\caption{Return-to-education interval, unconstrained versus conditional on race.}
\label{fig:conditional_race}
\end{figure}
\subsection{Information on \texorpdfstring{$\E[y^*\mid v]$}{E[y*|v]}}\label{subsec:conditional_v_num}

Age, binned into five cells, serves as the external variable $v$; it does
not enter $\tilde x$. Figure~\ref{fig:external_v} compares three nested
intervals for $\theta_{\mathrm{educ}}$: unconstrained,
$[-9.96,\ 32.82]$; known pooled mean, $[-8.91,\ 23.02]$; and conditional
on age, $[-8.53,\ 23.02]$ (see Proposition~\ref{prop:external_contraction}),
a further $1.2\%$ narrower than the pooled interval. The conditional
interval is contained in the pooled one, with the contraction
concentrated entirely at the lower endpoint (from $-8.91$ to $-8.53$);
the upper endpoint is identical. This is not a numerical coincidence:
Proposition~\ref{prop:external_contraction}'s equality condition holds
exactly at this endpoint's direction, since every age cell shares the
same optimal cutoff as the pooled problem
(\Cref{app:numerical}, Table~\ref{tab:cutoffs}), so the two support
values coincide by construction rather than by approximation.

\begin{figure}[htbp]
\centering
\includegraphics[width=0.6\textwidth]{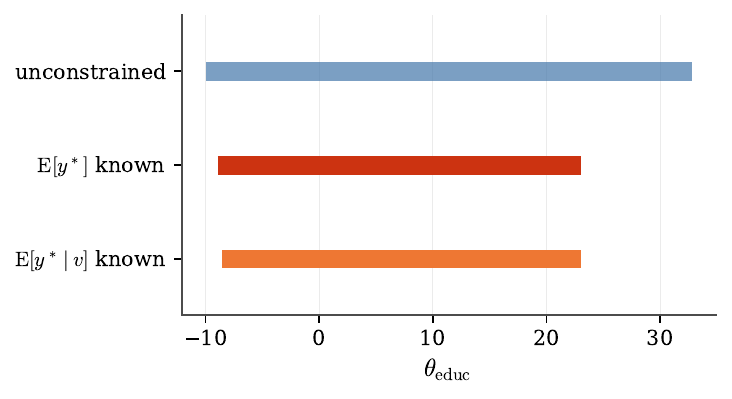}
\caption{Return-to-education interval: unconstrained, pooled mean known, conditional on age known.}
\label{fig:external_v}
\end{figure}

\section{Conclusion}\label{sec:conclusion}

This paper studies how auxiliary population moments sharpen identification of best linear predictors when the outcome is observed only through an interval. Using restricted selection sets, we show that the identifying content of auxiliary information varies sharply with the form of the restriction. A known unconditional mean intersects the original identified region with an affine hyperplane. A known conditional mean given a subvector of the regressors imposes $d_1+1$ exact linear restrictions and reduces the full region to an affine image of the remaining coefficient block. A conditional mean given an external variable further tightens the pooled-mean region by restricting how the available interval width may be allocated across values of that variable. Transformation moments have different effects depending on the target: a known mean of the transformed outcome imposes a hyperplane restriction on its own BLP, whereas the same moment used as auxiliary information about the original outcome generally narrows the region without reducing its dimension.

Across these cases, the value of auxiliary information is determined not merely by how many moments are known, but by how those moments restrict the allocation of latent outcomes within the observed intervals. For the mean and conditional-mean restrictions, we quantify this value through constrained allocation problems for interval width; for transformation restrictions, we provide geometric, support-function, and finite-sample characterizations. The CPS illustration shows that the resulting gains can be quantitatively important: conditioning the latent mean on race removes roughly one quarter of the identified width for the return to education, although the magnitude varies substantially with the information supplied.

Relative to \citet{BM2008}, whose sharp identification region for the interval-outcome BLP is the starting point here, the contribution is to show that a specific and commonly available kind of information --- population aggregates that already accompany coarsened data for confidentiality reasons --- has exploitable identifying content, and to give closed-form expressions for exactly how much. Coarsening is what creates the identification problem in the first place, but the same disclosure regime that requires coarsening typically also requires publishing exact aggregate invariants, and it is from those invariants, not from any external dataset, that the identifying power studied here is recovered. That said, this is a statement about the specific restrictions we study, not a general theory of how privacy regimes interact with partial identification; the results say what a mean, a conditional mean, or a transformation moment buys, not what an arbitrary disclosed statistic would buy.

The framework also has real limits, and they are not incidental to the results. Every closed form in \cref{sec:condexp} relies on $\tilde x$ being exactly observed, so that the identification problem reduces to a fixed Aumann integral rather than a $\theta$-dependent containment check; if the covariates are themselves interval-valued, as in \citet{BMM2011}, the machinery here does not directly apply. The quantification in \cref{subsec:value_of_mean} is tied to squared loss --- it is the knapsack representation, not the restricted-selection idea itself, that depends on this --- so the value of a restriction under quantile or other asymmetric losses is not covered by the formulas derived here. We also treat $\Theta^I$ and its restricted analogs purely as population objects; this paper does not establish how their sample analogs behave, so the CPS numbers should be read as an illustration of the geometry rather than as a fitted or standard-error-attached estimate. Restrictions are studied one at a time, and it is not obvious from anything shown here whether the contraction from two simultaneously imposed restrictions is additive, larger, or smaller, since the restrictions can overlap in which selections they each rule out. Finally, throughout the paper the disclosed aggregate is taken as exogenously given; nothing here says whether it was chosen well, or what a custodian trying to balance disclosure risk against downstream identifying power should release instead. Addressing any of these would change the analysis rather than extend it as a footnote, and we have deliberately left them open rather than sketch results we have not derived.

\noindent\textbf{Data Availability:} The data used in this paper are drawn from the Current Population Survey
Annual Social and Economic Supplement (CPS ASEC), a public-use survey
administered by the U.S. Census Bureau and the Bureau of Labor Statistics.

\noindent\textbf{Code:}
Code implementing every result below, together with the data extract and every figure (main text and appendix), is available at
\url{https://github.com/BehroozMoosavi/Codes/tree/main/PI_with_%20Auxiliary_restriction}.

\noindent\textbf{Funding:} This research received no specific grant from any funding agency in the
public, commercial, or not-for-profit sectors.

\bibliography{ref}

\appendix
\section{Background Results}\label{app:background}

This appendix collects the results from the theory of random sets and from
convex analysis that are used in the proofs of \cref{sec:condexp}: the Aumann
integral of a random set, Lyapunov's convexity theorem, the continuous fractional-knapsack (``bathtub'') problem, the population Frisch--Waugh--Lovell identity, and the segment representation of \citet{BM2008} that ties them together. Throughout, $(\Omega,\mathcal F,\PP)$
is a complete probability space and $\|\cdot\|$ is the Euclidean norm on
$\R^k$.

\subsection{Random sets, selections, and the Aumann integral}

\begin{definition}[Measurable correspondence and selections]\label{def:randomset}
A \emph{correspondence} (random set) $\Gamma:\Omega\rightrightarrows\R^k$
assigns to each $\omega$ a nonempty closed set
$\Gamma(\omega)\subseteq\R^k$. It is \emph{measurable} if for every open
$G\subseteq\R^k$,
\[
\{\omega:\Gamma(\omega)\cap G\neq\varnothing\}\in\mathcal F .
\]
A measurable map $f:\Omega\to\R^k$ with
$f(\omega)\in\Gamma(\omega)$ for $\PP$-a.e.\ $\omega$ is a
\emph{measurable selection}, and
\[
\Sel^1(\Gamma)
=
\left\{
f\in L^1(\Omega,\mathcal F,\PP;\R^k):
f(\omega)\in\Gamma(\omega)\ \PP\text{-a.s.}
\right\}
\]
is the set of \emph{integrable selections}. The correspondence is
\emph{integrably bounded} if there is $g\in L^1(\PP)$ with
\[
\sup_{z\in\Gamma(\omega)}\|z\|\leq g(\omega)
\qquad
\PP\text{-a.s.}
\]
\end{definition}

\begin{definition}[Aumann integral]\label{def:aumann}
The \emph{Aumann (selection) integral} of a measurable correspondence
$\Gamma$ is
\[
\int_\Omega\Gamma\,d\PP
=
\left\{
\E[f]:
f\in\Sel^1(\Gamma)
\right\}
\subseteq\R^k .
\]
\end{definition}

\begin{theorem}[Properties of the Aumann integral]\label{thm:aumann}
Let $\Gamma:\Omega\rightrightarrows\R^k$ be measurable, nonempty
compact-valued, and integrably bounded. Then:
\begin{enumerate}[label=\textup{(\roman*)}]
\item $\Sel^1(\Gamma)\neq\varnothing$, and
      $\int_\Omega\Gamma\,d\PP$ is a nonempty compact subset of $\R^k$;
\item if $\Gamma$ is convex-valued, then $\int_\Omega\Gamma\,d\PP$ is convex
      on any probability space;
\item if $(\Omega,\mathcal F,\PP)$ is atomless, then
      $\int_\Omega\Gamma\,d\PP$ is convex
      \textup{(}without any convexity of the values\textup{)}, and
      \[
      \int_\Omega\Gamma\,d\PP
      =
      \int_\Omega \overline{\mathrm{co}}\,\Gamma\,d\PP ;
      \]
\item the support function commutes with the integral: for every
      $u\in\R^k$,
      \begin{equation}\label{eq:support_commute}
      h_{\int_\Omega\Gamma\,d\PP}(u)
      =
      \int_\Omega h_{\Gamma(\omega)}(u)\,d\PP(\omega),
      \qquad
      h_{\Gamma(\omega)}(u)=\sup_{z\in\Gamma(\omega)}u'z .
      \end{equation}
\end{enumerate}
\end{theorem}

\begin{proof}
Nonemptiness of $\Sel^1(\Gamma)$ is the
Kuratowski--Ryll-Nardzewski measurable selection theorem: a measurable,
nonempty compact-valued correspondence admits a measurable selection, which is
integrable by integrable boundedness. The compactness in (i) and the
convexity in (iii), together with the identity
$\int\Gamma\,d\PP=\int\overline{\mathrm{co}}\,\Gamma\,d\PP$, are the
Richter--Aumann theorems for integrals of correspondences on atomless spaces; see \citet[Theorem~1.21 and Theorem~1.24]{Molchanov2005},
\citet{Richter1963}, and the econometric exposition in \citet{BM2008}.

For (ii), let $m_0,m_1\in\int\Gamma\,d\PP$ with selections
$f_0,f_1\in\Sel^1(\Gamma)$, and let $\lambda\in[0,1]$. Convex values give
\[
\lambda f_1(\omega)+(1-\lambda)f_0(\omega)\in\Gamma(\omega)
\]
a.s., so
\[
f_\lambda:=\lambda f_1+(1-\lambda)f_0
\]
belongs to $\Sel^1(\Gamma)$ and
\[
\E[f_\lambda]
=
\lambda m_1+(1-\lambda)m_0
\]
lies in the integral.

For (iv), the inequality ``$\leq$'' is immediate: for any selection $f$ and
any $u$,
\[
u'\E[f]=\E[u'f]\leq\E[h_{\Gamma}(u)].
\]
For ``$\geq$'', the map
\[
\omega\mapsto\operatorname{arg\,max}_{z\in\Gamma(\omega)}u'z
\]
admits a measurable selection $f_u$ by measurable selection applied to the
closed, integrably bounded face correspondence, and
\[
u'\E[f_u]=\E[h_\Gamma(u)].
\]
Hence the support function of the integral attains the right-hand side.
Finiteness follows from integrable boundedness.
\end{proof}

In every application in \cref{sec:condexp} the relevant correspondence has
the form
\[
\Gamma(\omega)
=
\{\tilde x(\omega)t:t\in[y_L(\omega),y_U(\omega)]\}
\]
or an augmentation thereof; it is automatically measurable, compact-valued,
and integrably bounded by
\[
\|\tilde x\|(|y_L|+|y_U|)
\]
under \cref{as:main}. Section~\ref{app:bm_construction} below records the
closed-form support function this segment structure implies.

\subsection{Lyapunov's convexity theorem}

\begin{theorem}[Lyapunov]\label{thm:lyapunov}
Let $\nu=(\nu_1,\dots,\nu_k)$ be a vector of finite signed measures on
$(\Omega,\mathcal F)$, each of which is atomless. Then the range
\[
R(\nu)
=
\left\{
\nu(A):
A\in\mathcal F
\right\}
\subseteq\R^k
\]
is compact and convex.
\end{theorem}

\begin{proof}
This is the classical theorem of \citet{Lyapunov1940}; a self-contained proof
via the extreme points of the weak-$*$ compact set of densities
$\{\tau:0\leq\tau\leq1\}$ is given in
\citet[Theorem~IV.10.5]{DunfordSchwartz1958}.
\end{proof}

The form used repeatedly in \cref{app:proofs} is the following measurable
``splitting'' lemma, which lets a single event simultaneously interpolate
finitely many integrals.

\begin{lemma}[Lyapunov splitting]
\label{lem:lyapunov_split}
Suppose $(\Omega,\mathcal F,\PP)$ is atomless.  If
$H\in L^1(\PP;\R^k)$, then for every $\lambda\in[0,1]$ there exists
$A_\lambda\in\mathcal F$ such that
\[
\int_{A_\lambda}H\,d\PP
=
\lambda\E[H].
\]
\end{lemma}

\begin{proof}
Define the finite vector measure
\[
\nu(A)=\int_AH\,d\PP.
\]
Since $\nu\ll\PP$ and $\PP$ is atomless, every component of $\nu$ is
nonatomic.  Lyapunov's convexity theorem implies that the range
\[
R(\nu)=\{\nu(A):A\in\mathcal F\}
\]
is compact and convex.  The range contains
\[
\nu(\varnothing)=0
\qquad\text{and}\qquad
\nu(\Omega)=\E[H].
\]
Therefore it contains the whole segment
$\{\lambda\E[H]:\lambda\in[0,1]\}$.  Hence for every $\lambda$ there is
$A_\lambda\in\mathcal F$ satisfying the desired equality.
\end{proof}

The lemma is used to splice two selections while interpolating finitely many
moments exactly.  If
\[
y_\lambda
=
y_1\mathbf 1_{A_\lambda}
+
y_0\mathbf 1_{A_\lambda^c},
\]
and $H=\Phi(y_1)-\Phi(y_0)$, then the choice of $A_\lambda$ gives
\[
\E[\Phi(y_\lambda)]
=
(1-\lambda)\E[\Phi(y_0)]
+
\lambda\E[\Phi(y_1)].
\]

\subsection{The continuous fractional knapsack}

\begin{lemma}[Continuous fractional knapsack and its dual]\label{lem:knapsack}
Let $\mu$ be a finite nonnegative measure on $(\Omega,\mathcal F)$, let
$Z\in L^1(\mu)$, and let $0\leq a\leq\mu(\Omega)$. Define the upper-tail
functional
\[
\mathcal T_a^\mu(Z)
=
\sup\left\{
\int_\Omega \tau Z\,d\mu:
\tau:\Omega\to[0,1]\text{ measurable},\
\int_\Omega \tau\,d\mu=a
\right\}.
\]
Then the following hold.
\begin{enumerate}[label=\textup{(\roman*)}]
\item \textup{(Duality)}
\begin{equation}\label{eq:knapsack_duality}
\mathcal T_a^\mu(Z)
=
\inf_{c\in\R}
\left\{
ca+\int_\Omega (Z-c)^+\,d\mu
\right\}.
\end{equation}

\item \textup{(Threshold solution)}
If $0<a<\mu(\Omega)$, there is a cutoff $c^\ast\in\R$ with
\[
\mu(Z>c^\ast)\leq a\leq\mu(Z\geq c^\ast),
\]
and the allocation
\[
\tau^\ast
=
\1\{Z>c^\ast\}
+
\gamma\,\1\{Z=c^\ast\},
\qquad
\gamma
=
\frac{a-\mu(Z>c^\ast)}{\mu(Z=c^\ast)}
\in[0,1],
\]
\textup{(}with any $\gamma\in[0,1]$ if $\mu(Z=c^\ast)=0$\textup{)}
is optimal. If $a=0$ the optimum is $\tau^\ast=0$, and if
$a=\mu(\Omega)$ it is $\tau^\ast=1$.

\item \textup{(Quantile form)}
Write $\Lambda=\mu(\Omega)>0$, $p=a/\Lambda$,
$\bar\mu=\mu/\Lambda$, and let $q_Z^{\bar\mu}$ be the generalized quantile of
$Z$ under $\bar\mu$. Then
\begin{equation}\label{eq:knapsack_quantile}
\mathcal T_a^\mu(Z)
=
\Lambda\int_{1-p}^{1} q_Z^{\bar\mu}(t)\,dt .
\end{equation}
\end{enumerate}
\end{lemma}

\begin{proof}
\emph{(i) Weak duality.} For any feasible $\tau$ and any $c\in\R$,
\[
\int \tau Z\,d\mu
=
c\int\tau\,d\mu+\int\tau(Z-c)\,d\mu
=
ca+\int\tau(Z-c)\,d\mu
\leq
ca+\int(Z-c)^+\,d\mu,
\]
because $0\leq\tau\leq1$ implies
\[
\tau(Z-c)\leq (Z-c)^+
\]
pointwise. Taking the supremum over feasible $\tau$ and then the infimum over
$c$ gives ``$\leq$''.

\emph{Attainment.} For $0<a<\Lambda$, the function
\[
c\mapsto\mu(Z>c)
\]
is nonincreasing and right-continuous with limits $\Lambda$ and $0$, so
\[
c^\ast=\inf\{c:\mu(Z>c)\leq a\}
\]
satisfies
\[
\mu(Z>c^\ast)\leq a\leq\mu(Z\geq c^\ast).
\]
The tie fraction $\gamma$ is chosen so that $\tau^\ast$ is feasible. Since
\[
\tau^\ast=1 \quad\text{on }\{Z>c^\ast\},
\qquad
\tau^\ast=0 \quad\text{on }\{Z<c^\ast\},
\]
we have
\[
(Z-c^\ast)\tau^\ast=(Z-c^\ast)^+
\qquad
\mu\text{-a.e.}
\]
Hence
\[
\int\tau^\ast Z\,d\mu
=
c^\ast a+\int(Z-c^\ast)^+\,d\mu,
\]
which matches the dual value at $c^\ast$. Together with weak duality, this
proves both (i) and the optimality of $\tau^\ast$ in (ii). The boundary cases
$a\in\{0,\Lambda\}$ are immediate, with the dual infimum attained in the
limit $c\to+\infty$ and $c\to-\infty$, respectively.

\emph{(iii) Quantile form.} Optimality of the threshold rule means the
supremum is obtained by assigning unit weight to the largest values of $Z$
until the budget $a$ is exhausted. Equivalently, under the normalized law
$\bar\mu$, this integrates the top $p=a/\Lambda$ fraction of the quantile
function, giving \eqref{eq:knapsack_quantile}. See
\citet{CormenLeisersonRivest89} for the finite/algorithmic version; the
measure-theoretic statement is the standard bathtub principle.
\end{proof}

\subsection{The Frisch--Waugh--Lovell identity}

\begin{lemma}[Population Frisch--Waugh--Lovell]\label{lem:fwl}
Partition the augmented regressor as
\[
\tilde x=(w',x_2')'
\]
with $w\in\R^{d_1+1}$ and $x_2\in\R^{d_2}$, and suppose
\[
Q=\E[\tilde x\tilde x']
\]
is positive definite. Write
\[
\Sigma_{11}=\E[ww'],
\quad
\Sigma_{12}=\E[wx_2']=\Sigma_{21}',
\quad
\Sigma_{22}=\E[x_2x_2'],
\]
and define the linear-projection residual and Schur complement
\[
x_2^{*}=x_2-\Sigma_{21}\Sigma_{11}^{-1}w,
\qquad
\Sigma_{2\cdot1}
=
\Sigma_{22}-\Sigma_{21}\Sigma_{11}^{-1}\Sigma_{12}.
\]
Then $\Sigma_{2\cdot1}$ is positive definite,
\[
\E[x_2^{*}w']=0,
\]
and for any random variable $m$ satisfying $\E[\|\tilde x\|\,|m|]<\infty$,
the solution
$\theta=(\theta_1',\theta_2')'$ of the population normal equations
\[
Q\theta=\E[\tilde x m]
\]
is given by
\begin{equation}\label{eq:fwl_identity}
\theta_2=\Sigma_{2\cdot1}^{-1}\,\E[x_2^{*}m],
\qquad
\theta_1=\Sigma_{11}^{-1}\bigl(\E[wm]-\Sigma_{12}\theta_2\bigr).
\end{equation}
\end{lemma}

\begin{proof}
Positive definiteness of $Q$ implies that of its principal block
$\Sigma_{11}$ and of the Schur complement $\Sigma_{2\cdot1}$ by standard
block-matrix algebra. Orthogonality is direct:
\[
\E[x_2^{*}w']
=
\Sigma_{21}
-
\Sigma_{21}\Sigma_{11}^{-1}\Sigma_{11}
=
0 .
\]

Since $\|w\|,\|x_2\|\leq\|\tilde x\|$, the hypothesis
$\E[\|\tilde x\|\,|m|]<\infty$ gives
\[
\E[\|w\|\,|m|]<\infty,
\qquad
\E[\|x_2\|\,|m|]<\infty,
\]
so $\E[wm]$ and $\E[x_2m]$ are well defined and finite; this is the only
integrability used below.

The normal equations in block form read
\[
\Sigma_{11}\theta_1+\Sigma_{12}\theta_2=\E[wm],
\qquad
\Sigma_{21}\theta_1+\Sigma_{22}\theta_2=\E[x_2m].
\]
Solving the first block for $\theta_1$ gives
\[
\theta_1
=
\Sigma_{11}^{-1}(\E[wm]-\Sigma_{12}\theta_2).
\]
Substituting into the second block,
\[
\Sigma_{21}\Sigma_{11}^{-1}
\bigl(\E[wm]-\Sigma_{12}\theta_2\bigr)
+
\Sigma_{22}\theta_2
=
\E[x_2m],
\]
that is,
\[
\bigl(\Sigma_{22}-\Sigma_{21}\Sigma_{11}^{-1}\Sigma_{12}\bigr)\theta_2
=
\E[x_2m]-\Sigma_{21}\Sigma_{11}^{-1}\E[wm]
=
\E[x_2^{*}m],
\]
where the last equality uses
\[
x_2^{*}=x_2-\Sigma_{21}\Sigma_{11}^{-1}w.
\]
Since $\Sigma_{2\cdot1}$ is invertible,
\[
\theta_2=\Sigma_{2\cdot1}^{-1}\E[x_2^{*}m],
\]
which is \eqref{eq:fwl_identity}. See \citet{Davidson2004} for the textbook
treatment.
\end{proof}
\subsection{The Beresteanu--Molinari segment representation and its support function}
\label{app:bm_construction}

The construction used throughout \cref{sec:condexp} -- representing the interval-outcome identification problem through the pointwise segment
\[
\Gamma(\omega)=\{\tilde x(\omega)t:t\in[y_L(\omega),y_U(\omega)]\}
\]
and recovering $\Theta^I$ as $Q^{-1}$ by applying $Q^{-1}$ to its Aumann integral -- is not new to this paper: it is a direct restatement, in the notation used here, of \citet[Proposition~4.1]{BM2008}, the result that introduced random-set methods into the identification analysis of best linear predictors with interval outcome data. We record the construction in general form, derive the closed-form support function it implies, and use the latter to state the computational algorithm suggested by this representation.

\paragraph{The identification problem.}
Absent additional information, a coefficient $\theta\in\R^{d+1}$ is an
admissible best linear predictor coefficient if and only if it solves the
population normal equations
\[
\E[\tilde x(y-\tilde x'\theta)]=\mathbf 0
\]
for \emph{some} completion $\eta$ of the joint distribution of
$(y,\tilde x)$ compatible with the observed interval bounds, i.e., for some
$y\in\Sel^1(Y)$. Written this way, verifying whether a candidate $\theta$ is
admissible requires searching over completions $\eta$ of the joint law of
$(y,\tilde x)$ subject to a continuum of conditional moment inequalities of
the form
\[
\eta\bigl([t_0,t_1]\times B\bigr)
\geq
\PP\bigl(y_L\geq t_0,\,
y_U\leq t_1,\,
\tilde x\in B\bigr),
\qquad
\forall\, t_0\leq t_1,\ \forall\, \text{measurable } B,
\]
by the Artstein-type characterization underlying \cref{thm:aumann}. This
description is sharp but, on its face, not tractable to compute with.
\citet{BM2008} showed that when $\tilde x$ is exactly observed, this
description is equivalent to equation \eqref{eq:condexp_base_id}: the two routes --
optimizing over completions of the joint law subject to the moment
inequalities above, versus taking the Aumann integral of the segment
$\Gamma$ -- characterize the same set $\Theta^I$. Finally, Proposition~\ref{prop:basic} in the main text establishes the properties of the the identification set in equation \eqref{eq:condexp_base_id} using our notation and
setting.

\paragraph{The support function of $\Theta^I$ in closed form.}
Combining \citet[Proposition~4.1]{BM2008} with \cref{thm:aumann}(iv) yields a
one-line alternative derivation of the directional bounds
\eqref{eq:unconstrained_upper}--\eqref{eq:unconstrained_lower}, in place of
the $\tau$-allocation argument used in \cref{sec:condexp}. For fixed
$\omega$ and $u\in\R^{d+1}$,
\[
h_{\Gamma(\omega)}(u)
=
\sup_{t\in[y_L(\omega),y_U(\omega)]}
(u'\tilde x(\omega))\,t
=
\bigl(u'\tilde x(\omega)\bigr)^+y_U(\omega)
-
\bigl(u'\tilde x(\omega)\bigr)^-y_L(\omega),
\]
because the segment is maximized at $y_U$ when $u'\tilde x>0$ and at $y_L$
when $u'\tilde x<0$. Since $\Theta^I=Q^{-1}\mathcal M$ and $Q^{-1}$ is
symmetric,
\[
h_{\Theta^I}(r)=h_{\mathcal M}(Q^{-1}r)
\]
for every $r\in\R^{d+1}$, and \cref{thm:aumann}(iv) gives
\[
h_{\mathcal M}(Q^{-1}r)
=
\E[h_{\Gamma(\omega)}(Q^{-1}r)].
\]
Writing
\[
s_r=r'Q^{-1}\tilde x
\]
as in \cref{sec:condexp},
\begin{equation}\label{eq:bm_support}
h_{\Theta^I}(r)
=
\E\bigl[(s_r)^+y_U\bigr]
-
\E\bigl[(s_r)^-y_L\bigr].
\end{equation}
A direct algebraic check, adding and subtracting
$\E[(s_r)^+y_L]$ and using
\[
s_ry_L=(s_r)^+y_L-(s_r)^-y_L,
\]
shows that \eqref{eq:bm_support} coincides exactly with
\eqref{eq:unconstrained_upper}. The two derivations -- the
$\tau$-allocation argument of \cref{sec:condexp} and the direct
support-function computation here -- are therefore dual routes to the same
closed form.

\paragraph{Three equivalent derivations.}
This coincidence is not an accident specific to this paper. The same closed-form directional bound has been obtained by at least three routes in the literature, each of which reappears somewhere in this paper's own argument: \citet{BM2008} via the Aumann-expectation/support-function route just given; \citet{Stoye2007} via direct constrained optimization over the selection $y\in\Sel^1(Y)$, which is exactly the $\tau$-allocation argument used to derive \eqref{eq:unconstrained_upper}--\eqref{eq:unconstrained_lower} in \cref{sec:condexp}; and \citet{MagnacMaurin2008} via an argument built on the Frisch--Waugh--Lovell theorem, the same identity recorded in \cref{lem:fwl} above and used to obtain the coordinate-breadth formula in Proposition~\ref{prop:width}. The agreement among these distinct approaches provides a useful cross-check on the formulas used in Propositions~\ref{prop:basic}--\ref{prop:width}.

\paragraph{The resulting algorithm.}
For compact convex identified sets, the support function fully determines the
set. In particular, Proposition~\ref{prop:basic} implies
\[
\Theta^I
=
\bigcap_{r\in S^d}
\bigl\{
\theta\in\R^{d+1}: r'\theta\leq h_{\Theta^I}(r)
\bigr\},
\qquad
S^d=\{r\in\R^{d+1}:\|r\|=1\}.
\]
The same support-function method applies to the restricted problems whenever
the corresponding attainable moment set is compact and convex. This suggests
the following direct construction, used to compute or estimate the objects in
\cref{sec:condexp}: (1) compute $Q$, or its sample analog
\[
Q_n=n^{-1}\sum_i \tilde x_i\tilde x_i';
\]
(2) for a chosen direction $r$, form the pointwise score
\[
s_r=r'Q^{-1}\tilde x,
\]
or
\[
s_{r,i}=r'Q_n^{-1}\tilde x_i
\]
at each observation; (3) evaluate the relevant closed-form expectation --
\eqref{eq:unconstrained_upper}/\eqref{eq:bm_support} for the unconstrained
set, the dual knapsack formula in  Proposition~\ref{prop:directional_contraction_mean} for the mean-constrained set, and so on; (4) repeat over the $2(d+1)$ directions $\pm e_j$, $j=0,..,d$, for coordinate bounds, or over a finer grid on $S^d$ to trace an outer polytope approximation to the boundary of $\Theta^I$. Every step is a closed-form expectation rather than the solution to a nonlinear or convex program, which is what makes the restrictions studied in \cref{sec:condexp} directly estimable by their sample analogs. Formal statistical inference for the resulting estimator -- establishing that the sample support function converges to a Gaussian process and that this can be used to build confidence statements for $\Theta^I$ or for functionals of it -- is developed in \citet{BM2008} and, for settings in which the bounding functions themselves must be estimated, in \citet{ChandrasekharEtAl2019}; this lies outside the identification-only scope of the present paper.

\section{Proofs for \texorpdfstring{\cref{sec:condexp}}{Section~\ref{sec:condexp}}}\label{app:proofs}

Throughout, we use the Aumann-integral theorem, Lyapunov's convexity theorem,
the continuous fractional-knapsack theorem, and the population
Frisch--Waugh--Lovell identity collected in \cref{app:background}.

\subsection{Proof of Proposition \ref{prop:basic}}

\begin{proof}
Nonemptiness follows because $y_L\in\Sel^1(Y)$. Indeed, since the first
coordinate of $\tilde x$ is one, Assumption \ref{as:main} implies
\[
\E[|y_L|+|y_U|]<\infty .
\]

Define the random correspondence
\[
\Gamma(\omega)
=
\left\{
\tilde x(\omega)t:
t\in[y_L(\omega),y_U(\omega)]
\right\}
\subseteq\R^{d+1}.
\]
It is measurable, nonempty, compact-valued, and convex-valued. Moreover,
\[
\sup_{z\in\Gamma(\omega)}\|z\|
\leq
\|\tilde x(\omega)\|
\bigl(|y_L(\omega)|+|y_U(\omega)|\bigr),
\]
whose expectation is finite under Assumption \ref{as:main}. Thus $\Gamma$ is
integrably bounded.

Its integrable selections are exactly the vectors $\tilde x y$ with
$y\in\Sel^1(Y)$. Therefore,
\[
\int_\Omega\Gamma\,d\PP
=
\left\{
\E[\tilde x y]:
y\in\Sel^1(Y)
\right\}
=
\mathcal M.
\]
The Aumann integral of an integrably bounded measurable compact-valued
correspondence in a finite-dimensional space is compact. Since $\Gamma$ is
convex-valued, its integral is convex. Hence $\mathcal M$ is nonempty,
compact, and convex.

Because $Q^{-1}$ is an invertible linear map,
\[
\Theta^I=Q^{-1}\mathcal M
\]
is also nonempty, compact, and convex.
\end{proof}

\subsection{Proof of Proposition \ref{prop:mean_hyperplane}}

\begin{proof}
Suppose $\theta\in\Theta^I_\kappa$. Then there exists
$y\in\Sel^1(Y\mid\kappa)$ such that
\[
Q\theta=\E[\tilde x y].
\]
Hence $\theta\in\Theta^I$.

Let
\[
e_0=(1,0,\dots,0)'.
\]
Since the first coordinate of $\tilde x$ is one,
\[
e_0'Q
=
\E[\tilde x]'.
\]
Therefore,
\[
\E[\tilde x]'\theta
=
e_0'Q\theta
=
e_0'\E[\tilde x y]
=
\E[y]
=
\kappa.
\]
Thus $\theta\in\Theta_\kappa$, proving
\[
\Theta^I_\kappa
\subseteq
\Theta^I\cap\Theta_\kappa.
\]

Conversely, suppose
\[
\theta\in\Theta^I\cap\Theta_\kappa.
\]
There exists $y\in\Sel^1(Y)$ such that
\[
Q\theta=\E[\tilde x y].
\]
Taking first coordinates,
\[
\E[y]
=
e_0'Q\theta
=
\E[\tilde x]'\theta
=
\kappa.
\]
Hence $y\in\Sel^1(Y\mid\kappa)$ and
$\theta\in\Theta^I_\kappa$.

This proves
\[
\Theta^I_\kappa=\Theta^I\cap\Theta_\kappa.
\]

The set $\Theta^I$ is compact and convex, and $\Theta_\kappa$ is a closed
affine hyperplane. Their intersection is compact and convex. Nonemptiness
follows from compatibility.

Since $\Theta_\kappa$ has dimension $d$,
\[
\dim\operatorname{aff}(\Theta^I_\kappa)\leq d.
\]
If $\Theta^I$ is full-dimensional and its relative interior intersects
$\Theta_\kappa$, the intersection contains a relatively open subset of
$\Theta_\kappa$, and therefore its affine hull is $\Theta_\kappa$.
\end{proof}

\subsection{Proof of Proposition \ref{prop:width}}

\begin{proof}
Using equations \eqref{eq:unconstrained_upper} and
\eqref{eq:unconstrained_lower},
\begin{align*}
w_{\Theta^I}(r)
&=
\E[(s_r)^+\Delta]+\E[(s_r)^-\Delta]\\
&=
\E[|s_r|\Delta].
\end{align*}

For the coordinate identity, define
\[
\gamma_j'
=
\E[\tilde x_j\tilde x_{-j}']
\E[\tilde x_{-j}\tilde x_{-j}']^{-1},
\qquad
\tilde x_j^*
=
\tilde x_j-\gamma_j'\tilde x_{-j}.
\]
Then
\[
\E[\tilde x_{-j}\tilde x_j^*]=0.
\]
Let
\[
g
=
\frac{\tilde x_j^*}{\sigma_j^2}.
\]
For every $k\neq j$,
\[
\E[\tilde x_k g]=0,
\]
while
\[
\E[\tilde x_jg]
=
\frac{\E[\tilde x_j\tilde x_j^*]}{\sigma_j^2}.
\]
Since
\[
\tilde x_j=\tilde x_j^*+\gamma_j'\tilde x_{-j}
\]
and $\E[\tilde x_{-j}\tilde x_j^*]=0$, we have
\[
\E[\tilde x_j\tilde x_j^*]
=
\E[(\tilde x_j^*)^2]
=
\sigma_j^2.
\]
Therefore,
\[
\E[\tilde x_jg]=1.
\]
Hence
\[
\E[\tilde x g]=e_j.
\]

Since $g$ is linear in $\tilde x$, write $g=a'\tilde x$. Then
\[
\E[\tilde x g]
=
\E[\tilde x\tilde x']a
=
Qa.
\]
Thus
\[
Qa=e_j,
\]
so
\[
a=Q^{-1}e_j.
\]
Therefore,
\[
g=e_j'Q^{-1}\tilde x=s_j.
\]
Substitution into \eqref{eq:dir_width} gives
\[
\theta_j^{\max}-\theta_j^{\min}
=
\frac{\E[|\tilde x_j^*|\Delta]}{\sigma_j^2}.
\]
Since $\sigma_j^2=\E[(\tilde x_j^*)^2]$, this is exactly
\eqref{eq:coord_width}.
\end{proof}

\subsection{Proof of Proposition \ref{prop:directional_contraction_mean} and Corollary \ref{cor:local_mean_contraction}}

\begin{proof}
Every $y\in\Sel^1(Y\mid\kappa)$ can be written as
\[
y=y_L+\tau\Delta
\]
for some measurable $\tau:\Omega\to[0,1]$, and the mean restriction is
equivalent to
\[
\E[\tau\Delta]=\kappa-\E[y_L]=\alpha.
\]
Therefore,
\[
U_\kappa(r)
:=
\sup_{\theta\in\Theta^I_\kappa}r'\theta
=
\E[s_r y_L]
+
\sup_{\tau}
\left\{
\E[s_r\tau\Delta]:
0\leq\tau\leq1,\ \E[\tau\Delta]=\alpha
\right\}.
\]

Let $\Lambda=\E[\Delta]>0$ and
\[
\mu(A)=\E[\Delta\1_A],
\qquad
\bar\mu(A)=\frac{\mu(A)}{\Lambda},
\]
so that $\mu$ is the finite measure of \cref{lem:knapsack} with
$\mu(\Omega)=\Lambda$, and $\alpha\in[0,\Lambda]$ by Assumption
\ref{assm:fixed_mean}. Applying the quantile form of the continuous
fractional-knapsack lemma, \cref{lem:knapsack}(iii), with $Z=s_r$ and
$a=\alpha$, gives directly
\[
\sup_{\substack{0\leq\tau\leq1\\
\E[\tau\Delta]=\alpha}}
\E[s_r\tau\Delta]
=
\Lambda\int_{1-p_\kappa}^{1}q_r(t)\,dt,
\]
where $p_\kappa=\alpha/\Lambda$ and $q_r$ is the generalized quantile
function of $s_r$ under $\bar\mu$. Hence
\[
U_\kappa(r)
=
\E[s_ry_L]
+
\Lambda\int_{1-p_\kappa}^{1}q_r(t)\,dt,
\]
which is the constrained upper support value stated in the proposition.

The unconstrained upper support value admits the same representation. By
\eqref{eq:unconstrained_upper}, the unconstrained maximizer allocates
width to exactly the states with $s_r>0$, so, writing
$u_r=\bar\mu(s_r>0)$, applying \cref{lem:knapsack}(iii) at the budget
$a=u_r\Lambda=\E[\Delta\1\{s_r>0\}]$ -- for which the optimal threshold
in part (ii) of the lemma is exactly $c^\ast=0$ -- gives
\[
\E[(s_r)^+\Delta]
=
\Lambda\int_{1-u_r}^{1}q_r(t)\,dt,
\]
and therefore
\[
\sup_{\theta\in\Theta^I}r'\theta
=
\E[s_ry_L]+\E[(s_r)^+\Delta]
=
\E[s_ry_L]+\Lambda\int_{1-u_r}^{1}q_r(t)\,dt.
\]

Therefore,
\[
\delta_\kappa(r)
=
\Lambda
\left[
\int_{1-u_r}^{1}q_r(t)\,dt
-
\int_{1-p_\kappa}^{1}q_r(t)\,dt
\right].
\]
Equivalently,
\[
\delta_\kappa(r)
=
\Lambda
\int_{\min\{1-p_\kappa,\,1-u_r\}}^{\max\{1-p_\kappa,\,1-u_r\}}
|q_r(t)|\,dt.
\]
Indeed, the quantile is nonpositive to the left of the zero-quantile region
and nonnegative to its right. This proves equation \eqref{eq:mean_band_integral} in the proposition.

For the breadth identity, note that
\[
w_{\Theta^I}(r)
=
\sup_{\theta\in\Theta^I}r'\theta
-
\inf_{\theta\in\Theta^I}r'\theta
=
\sup_{\theta\in\Theta^I}r'\theta
+
\sup_{\theta\in\Theta^I}(-r)'\theta,
\]
and the same identity holds for $\Theta^I_\kappa$. Hence
\[
w_{\Theta^I}(r)-w_{\Theta^I_\kappa}(r)
=
\delta_\kappa(r)+\delta_\kappa(-r),
\]
which proves equation \eqref{eq:breadth_contraction_mean} in the proposition.

Now suppose the law of $s_r$ under $\bar\mu$ has a density
$f_r$ that is continuous and strictly positive at zero. Then
\[
q_r(1-u_r)=0
\]
and, locally around $1-u_r$,
\[
q_r(t)
=
\frac{t-(1-u_r)}{f_r(0)}
+
o\!\left(|t-(1-u_r)|\right).
\]
Integrating this expansion over an interval of length
$|p_\kappa-u_r|$ gives
\[
\delta_\kappa(r)
=
\frac{\Lambda}{2f_r(0)}
(p_\kappa-u_r)^2
+
o\!\left((p_\kappa-u_r)^2\right).
\]
This proves Corollary \ref{cor:local_mean_contraction}.
\end{proof}

\subsection{Proof of Proposition \ref{prop:transformed_convexity}}

\begin{proof}
Take $m_0,m_1\in\mathcal M_f(\kappa)$, generated by
$y_0,y_1\in\Sel^1(Y\mid\kappa)$. Define
\[
H
=
\begin{pmatrix}
y_1-y_0\\[2mm]
\tilde x\{f(y_1)-f(y_0)\}
\end{pmatrix}.
\]
The vector $H$ is integrable under Assumptions \ref{as:main} and
\ref{as:int_f}. Indeed, $y_0$ and $y_1$ lie between $y_L$ and $y_U$, and $|y_1-y_0| \leq 2|y_L|+|y_U|$, whose expectation is finite from Assumption \ref{as:main} because the intercept is included. Assumption \ref{as:int_f} gives
\[
|f(y_j)|\leq F,
\qquad j=0,1,
\]
with $\E[\|\tilde x\|F]<\infty$.

Define
\[
\nu(A)=\int_AH\,d\PP.
\]
Because $\nu\ll\PP$ and $\PP$ is atomless, Lyapunov's theorem implies that
the range of $\nu$ is compact and convex. Hence, for every
$\lambda\in[0,1]$, there exists $A_\lambda$ such that
\[
\nu(A_\lambda)=\lambda\nu(\Omega).
\]

Define
\[
y_\lambda
=
y_1\1_{A_\lambda}+y_0\1_{A_\lambda^c}.
\]
Then $y_\lambda\in\Sel^1(Y)$. The first coordinate of the vector-measure
identity gives
\[
\E[y_\lambda]
=
(1-\lambda)\E[y_0]+\lambda\E[y_1]
=
\kappa.
\]
Thus $y_\lambda\in\Sel^1(Y\mid\kappa)$.

The remaining coordinates give
\[
\E[\tilde x f(y_\lambda)]
=
(1-\lambda)m_0+\lambda m_1.
\]
Therefore $\mathcal M_f(\kappa)$ is convex. Its image under $Q^{-1}$ is also
convex.
\end{proof}
\subsection{Proof of Proposition \ref{prop:moment_aux_convexity}}

\begin{proof}

\textbf{Integrability.} For any $y_0,y_1\in\Sel^1(Y)$, both $f(y_0)$ and
$f(y_1)$ are integrable: Assumption~\ref{as:int_f} gives $|f(y_i)|\leq F$
with $\E[F^2]<\infty$, and $\E[F^2]<\infty\Rightarrow\E[F]<\infty$ since
$F\geq0$, so $f(y_i)\in L^1(\PP)$. Likewise $\tilde x y_0,\tilde x
y_1\in L^1(\PP;\R^{d+1})$: since $y_L\leq y_i\leq y_U$ a.s.,
$|y_i|\leq|y_L|+|y_U|$, and Assumption~\ref{as:main} gives
$\E[\|\tilde x\|(|y_L|+|y_U|)]<\infty$, so $\E[\|\tilde x\||y_i|]<\infty$
for $i=0,1$.

\smallskip
\noindent\textbf{Step 1: $\mathcal K_f$ is convex.}
Fix $\kappa_0,\kappa_1\in\mathcal K_f$ and $\lambda\in[0,1]$; we show
$(1-\lambda)\kappa_0+\lambda\kappa_1\in\mathcal K_f$. By definition there
exist $y_0,y_1\in\Sel^1(Y)$ with $\E[f(y_0)]=\kappa_0$ and
$\E[f(y_1)]=\kappa_1$. Define the scalar random variable
\[
H=f(y_1)-f(y_0)\in L^1(\PP)
\]
and the vector measure
\[
\nu(A)=\int_A H\,d\PP,
\qquad A\in\mathcal F.
\]
Since $(\Omega,\mathcal F,\PP)$ is atomless (Assumption~\ref{as:atomless})
and $\nu$ is absolutely continuous with respect to $\PP$ with integrable
density $H$, $\nu$ is a nonatomic, finite (one-dimensional) vector
measure. By the Lyapunov convexity theorem for vector measures
(\cref{app:background}), the range $\{\nu(A):A\in\mathcal F\}$ is convex
and compact. Since $\nu(\varnothing)=0$ and $\nu(\Omega)=\E[H]$ both lie
in this range, convexity of the range implies that for every
$\lambda\in[0,1]$ there exists $A_\lambda\in\mathcal F$ with
\[
\nu(A_\lambda)=\lambda\,\nu(\Omega)
=
\lambda\,\E[f(y_1)-f(y_0)].
\]
Define
\[
y_\lambda=y_1\mathbf 1_{A_\lambda}+y_0\mathbf 1_{A_\lambda^c}.
\]
Since $y_0,y_1\in\Sel^1(Y)$ and $y_L\leq y_\lambda\leq y_U$ pointwise
(each state simply follows whichever of $y_0,y_1$ it is assigned to,
both of which satisfy the interval constraint), $y_\lambda\in\Sel^1(Y)$.
Moreover,
\begin{align*}
\E[f(y_\lambda)]
&=
\E\bigl[f(y_1)\mathbf 1_{A_\lambda}\bigr]
+
\E\bigl[f(y_0)\mathbf 1_{A_\lambda^c}\bigr]
\\&=
\E[f(y_0)]
+
\E\bigl[(f(y_1)-f(y_0))\mathbf 1_{A_\lambda}\bigr]
\\&=
\kappa_0+\nu(A_\lambda)
\\&=
\kappa_0+\lambda(\kappa_1-\kappa_0)
=
(1-\lambda)\kappa_0+\lambda\kappa_1.
\end{align*}
Hence $(1-\lambda)\kappa_0+\lambda\kappa_1\in\mathcal K_f$, proving
$\mathcal K_f$ is convex.

\smallskip
\noindent\textbf{Step 2: $\Theta^I_{\kappa_f}$ is nonempty.} Since
$\kappa_f\in\mathcal K_f$ by hypothesis, there exists
$y\in\Sel^1(Y)$ with $\E[f(y)]=\kappa_f$, i.e.,
$y\in\Sel^1(Y\mid f,\kappa_f)$, so $Q^{-1}\E[\tilde xy]\in\Theta^I_{\kappa_f}$.

\smallskip
\noindent\textbf{Step 3: $\Theta^I_{\kappa_f}$ is convex.}
Fix $\theta_0,\theta_1\in\Theta^I_{\kappa_f}$ and $\lambda\in[0,1]$; we
exhibit a point of $\Theta^I_{\kappa_f}$ equal to
$(1-\lambda)\theta_0+\lambda\theta_1$. By definition there exist
$y_0,y_1\in\Sel^1(Y\mid f,\kappa_f)$ with $Q^{-1}\E[\tilde xy_0]=\theta_0$
and $Q^{-1}\E[\tilde xy_1]=\theta_1$; in particular
$\E[f(y_0)]=\E[f(y_1)]=\kappa_f$.

Define the $(d+2)$-dimensional integrable random vector
\[
H=
\begin{pmatrix}
f(y_1)-f(y_0)\\
\tilde x(y_1-y_0)
\end{pmatrix}
\in L^1(\PP;\R^{d+2})
\]
(integrability of each component was verified above) and the
corresponding vector measure $\nu(A)=\int_AH\,d\PP$. As in Step 1,
atomlessness of $\PP$ makes $\nu$ a nonatomic finite vector measure, so
by the Lyapunov convexity theorem its range is convex and compact, and
for every $\lambda\in[0,1]$ there exists $A_\lambda\in\mathcal F$ with
\[
\nu(A_\lambda)=\lambda\,\nu(\Omega).
\]
Because $y_0,y_1\in\Sel^1(Y\mid f,\kappa_f)$, the \emph{first} coordinate
of $\nu(\Omega)$ is
\[
\E[f(y_1)]-\E[f(y_0)]=\kappa_f-\kappa_f=0,
\]
so the first coordinate of $\nu(A_\lambda)=\lambda\nu(\Omega)$ is $0$ for
\emph{every} $\lambda\in[0,1]$ -- not merely at $\lambda=0,1$. This is
the key point at which the argument differs from Step 1: the constraint
coordinate is pinned to zero along the entire splice, simultaneously
with the $\theta$-coordinates tracing out the segment.

Define $y_\lambda=y_1\mathbf 1_{A_\lambda}+y_0\mathbf 1_{A_\lambda^c}$,
which lies in $\Sel^1(Y)$ by the same pointwise argument as in Step 1.
Then, exactly as in Step 1's computation applied to the first coordinate
of $H$,
\[
\E[f(y_\lambda)]
=
\E[f(y_0)]+(\text{first coordinate of }\nu(A_\lambda))
=
\kappa_f+0=\kappa_f,
\]
so $y_\lambda\in\Sel^1(Y\mid f,\kappa_f)$ for every $\lambda\in[0,1]$.
Applying the same computation to the remaining $d+1$ coordinates of $H$,
\begin{align*}
\E[\tilde xy_\lambda]
&=
\E[\tilde xy_0]+(\text{remaining coordinates of }\nu(A_\lambda))
\\&=
\E[\tilde xy_0]+\lambda\bigl(\E[\tilde xy_1]-\E[\tilde xy_0]\bigr)
\\&=
(1-\lambda)\E[\tilde xy_0]+\lambda\E[\tilde xy_1].
\end{align*}
Applying the linear map $Q^{-1}$,
\[
Q^{-1}\E[\tilde xy_\lambda]
=
(1-\lambda)Q^{-1}\E[\tilde xy_0]+\lambda Q^{-1}\E[\tilde xy_1]
=
(1-\lambda)\theta_0+\lambda\theta_1.
\]
Since $y_\lambda\in\Sel^1(Y\mid f,\kappa_f)$, the left-hand side is a
point of $\Theta^I_{\kappa_f}$ by definition~\eqref{eq:moment_aux_set}.
Hence $(1-\lambda)\theta_0+\lambda\theta_1\in\Theta^I_{\kappa_f}$ for
every $\lambda\in[0,1]$, proving $\Theta^I_{\kappa_f}$ is convex.
\end{proof}

\subsection{Extensions: transformations and quantile information}
\label{app:extensions}

This subsection records the transformation and quantile extensions referred
to in Section~\ref{subsec:transformations}.  The core results in
Section~\ref{sec:condexp} require only Assumption~\ref{as:main}, together
with the stated auxiliary information.  The general transformation and
quantile machinery below requires additional regularity because nonlinear
and discontinuous transformations do not preserve the simple linear
knapsack structure of the untransformed problem.

Let
\[
f=(f_1,\dots,f_J)':\R\to\R^J
\]
be a vector of transformations, and let $\kappa_f\in\R^J$ be a known vector
of population transformation moments.

\begin{assumption}[Regular transformation for the appendix]
\label{as:regular_transform}
Each $f_j$ is continuous.  Moreover, there exists a nonnegative random
variable $F\in L^1(\PP)$ such that
\[
\sup_{t\in[y_L,y_U]}\|f(t)\|\leq F
\qquad
\PP\text{-a.s.}
\]
\end{assumption}

Define
\[
\mathcal K_f
=
\left\{
\E[f(y)]:
y\in\Sel^1(Y)
\right\}
\subseteq\R^J,
\]
and, for $\kappa_f\in\mathcal K_f$,
\[
\Sel^1(Y\mid f,\kappa_f)
=
\left\{
y\in\Sel^1(Y):
\E[f(y)]=\kappa_f
\right\}.
\]
The corresponding sharp region for the original BLP coefficient of $y$ is
\[
\Theta^I_{y\mid f,\kappa_f}
=
\left\{
Q^{-1}\E[\tilde x y]:
y\in\Sel^1(Y\mid f,\kappa_f)
\right\}.
\]
This differs from the transformed-outcome region
$\Theta^I_{f\mid\kappa_f}$ in Section~\ref{subsec:transformations}, where
the target is the BLP coefficient of $f(y^*)$.  Here the transformation
moment is auxiliary information imposed on the original latent outcome.

\begin{proposition}[Convexity under regular transformation moments]
\label{prop:appendix_transform_convexity}
Suppose Assumptions~\ref{as:main}, \ref{as:atomless}, and
\ref{as:regular_transform} hold.  Then $\mathcal K_f$ is compact and convex.
If $\kappa_f\in\mathcal K_f$, then $\Theta^I_{y\mid f,\kappa_f}$ is
nonempty, compact, and convex.
\end{proposition}

\begin{proof}
Compactness of $\mathcal K_f$ follows from Theorem ~\ref{thm:aumann}
applied to the compact-valued correspondence
\[
\omega\mapsto
\{f(t):t\in[y_L(\omega),y_U(\omega)]\}.
\]
Continuity of $f$ makes the values compact, and
Assumption~\ref{as:regular_transform} gives integrable boundedness.

For convexity, take $y_0,y_1\in\Sel^1(Y)$.  Let
\[
H=f(y_1)-f(y_0)\in L^1(\PP;\R^J).
\]
By Lemma~\ref{lem:lyapunov_split}, for every $\lambda\in[0,1]$ there exists
$A_\lambda\in\mathcal F$ such that
\[
\int_{A_\lambda}H\,d\PP=\lambda\E[H].
\]
Define
\[
y_\lambda
=
y_1\mathbf 1_{A_\lambda}
+
y_0\mathbf 1_{A_\lambda^c}.
\]
Then $y_\lambda\in\Sel^1(Y)$ and
\[
\E[f(y_\lambda)]
=
(1-\lambda)\E[f(y_0)]
+
\lambda\E[f(y_1)].
\]
Hence $\mathcal K_f$ is convex.

Now take $y_0,y_1\in\Sel^1(Y\mid f,\kappa_f)$.  Define the stacked vector
\[
H
=
\begin{pmatrix}
f(y_1)-f(y_0)\\
\tilde x(y_1-y_0)
\end{pmatrix}.
\]
It is integrable under Assumptions~\ref{as:main} and
\ref{as:regular_transform}.  Applying Lemma~\ref{lem:lyapunov_split} gives a
splice $y_\lambda$ satisfying
\[
\E[f(y_\lambda)]
=
(1-\lambda)\kappa_f+\lambda\kappa_f
=
\kappa_f,
\]
and
\[
\E[\tilde x y_\lambda]
=
(1-\lambda)\E[\tilde x y_0]
+
\lambda\E[\tilde x y_1].
\]
Thus the attainable cross-moment set under the transformation restriction is
convex, and its linear image under $Q^{-1}$ is convex.  Compactness follows
from the Aumann integral of the compact-valued correspondence
\[
\omega\mapsto
\left\{
(f(t)',\tilde x(\omega)'t)':
t\in[y_L(\omega),y_U(\omega)]
\right\}.
\]
Its image under $Q^{-1}$ is compact as well.
\end{proof}

\begin{assumption}[Relative-interior constraint qualification]
\label{as:ri_cq}
The known moment $\kappa_f$ belongs to the relative interior of
$\mathcal K_f$.
\end{assumption}

\begin{proposition}[Lagrangian support formula]
\label{prop:appendix_lagrangian}
Suppose Assumptions~\ref{as:main}, \ref{as:atomless},
\ref{as:regular_transform}, and \ref{as:ri_cq} hold.  Fix
$r\in\R^{d+1}$ and write
\[
s_r=r'Q^{-1}\tilde x.
\]
Then the support function of
$\Theta^I_{y\mid f,\kappa_f}$ is
\[
h_{\Theta^I_{y\mid f,\kappa_f}}(r)
=
\min_{c\in\R^J}
\left\{
c'\kappa_f
+
\E\left[
\sup_{t\in[y_L,y_U]}
\{s_rt-c'f(t)\}
\right]
\right\}.
\]
The minimum is attained.
\end{proposition}

\begin{proof}
For weak duality, take any $y\in\Sel^1(Y\mid f,\kappa_f)$ and any
$c\in\R^J$.  Pointwise,
\[
s_ry-c'f(y)
\leq
\sup_{t\in[y_L,y_U]}\{s_rt-c'f(t)\}.
\]
Taking expectations and using $\E[f(y)]=\kappa_f$ gives
\[
\E[s_ry]
\leq
c'\kappa_f+
\E\left[
\sup_{t\in[y_L,y_U]}\{s_rt-c'f(t)\}
\right].
\]
Taking the supremum over feasible $y$ and then the infimum over $c$ gives
one inequality.

For the reverse inequality, consider the finite-dimensional attainable set
\[
\mathcal A_r
=
\left\{
\left(\E[f(y)],\E[s_ry]\right):
y\in\Sel^1(Y)
\right\}
\subseteq\R^{J+1}.
\]
By the same Aumann and Lyapunov-splicing arguments used above,
$\mathcal A_r$ is compact and convex.  Let
\[
V(\kappa_f)
=
\sup\{z:(\kappa_f,z)\in\mathcal A_r\}.
\]
Assumption~\ref{as:ri_cq} places $\kappa_f$ in the relative interior of the
projection of $\mathcal A_r$ onto its first $J$ coordinates.  Hence the
closed convex set $\mathcal A_r$ admits a supporting hyperplane at
$(\kappa_f,V(\kappa_f))$ with a nonzero coefficient on the last coordinate.
After normalizing this last coefficient to one, there exists
$c^*\in\R^J$ such that, for all $(b,z)\in\mathcal A_r$,
\[
z
\leq
V(\kappa_f)+c^{*\prime}(b-\kappa_f).
\]
Equivalently, for every $y\in\Sel^1(Y)$,
\[
\E[s_ry]
\leq
V(\kappa_f)
+
c^{*\prime}\{\E[f(y)]-\kappa_f\}.
\]
Rearranging gives
\[
V(\kappa_f)
\geq
c^{*\prime}\kappa_f
+
\E[s_ry-c^{*\prime}f(y)]
\]
for every selection $y$.  Taking the supremum over selections and applying
the support-function identity to the compact-valued correspondence
\[
\omega\mapsto
\left\{
s_r(\omega)t-c^{*\prime}f(t):
t\in[y_L(\omega),y_U(\omega)]
\right\}
\]
yields
\[
V(\kappa_f)
\geq
c^{*\prime}\kappa_f
+
\E\left[
\sup_{t\in[y_L,y_U]}
\{s_rt-c^{*\prime}f(t)\}
\right].
\]
Together with weak duality, this proves equality and shows that $c^*$ is a
minimizer.
\end{proof}

\subsubsection{Quantile information}

Quantile information is obtained by taking an indicator transformation, for
example
\[
f(t)=\mathbf 1\{t\leq m\}.
\]
Then the restriction
\[
\E[\mathbf 1\{y^*\leq m\}]=q
\]
states that the latent distribution places probability $q$ weakly below
$m$.  Since the indicator is discontinuous, it does not satisfy
Assumption~\ref{as:regular_transform}.  The following finite-sample
description gives the support value directly by sorting.

\begin{proposition}[Finite-sample quantile restriction]
\label{prop:appendix_quantile}
Consider observations
\[
(y_{L,i},y_{U,i},\tilde x_i),
\qquad i=1,\dots,n,
\]
and fix a direction $r$.  Let
\[
Q_n=n^{-1}\sum_{i=1}^n\tilde x_i\tilde x_i',
\qquad
s_i=r'Q_n^{-1}\tilde x_i.
\]
Impose the empirical quantile restriction
\[
n^{-1}\sum_{i=1}^n\mathbf 1\{y_i\leq m\}=q.
\]
Define
\[
B=\{i:y_{U,i}\leq m\},
\qquad
A=\{i:y_{L,i}>m\},
\qquad
S=\{i:y_{L,i}\leq m<y_{U,i}\}.
\]
Suppose
\[
K=qn-|B|
\]
is an integer satisfying $0\leq K\leq |S|$.  For each $i\in S$, define
\[
t_i^-=
\begin{cases}
m, & s_i>0,\\
y_{L,i}, & s_i\leq0,
\end{cases}
\qquad
t_i^+=
\begin{cases}
y_{U,i}, & s_i>0,\\
m, & s_i\leq0,
\end{cases}
\]
and
\[
g_i=s_i(t_i^- - t_i^+).
\]
Then the supremum of the directional objective under the empirical quantile restriction is obtained by assigning the ``below'' status to the $K$ indices in $S$ with the largest values of $g_i$.  The value is interpreted as a supremum: when $s_i\leq0$ and an index is assigned to the above group, the choice $t_i^+=m$ may need to be approached from above by $m+\varepsilon$ rather than attained exactly.
\end{proposition}

\begin{proof}
For $i\in B$, every admissible value satisfies $y_i\leq m$.  For
$i\in A$, every admissible value satisfies $y_i>m$. Hence only the
indices in $S$ have a nontrivial below/above assignment.

Conditional on assigning $i\in S$ to the below group, the best representative
is $t_i^-$.  Conditional on assigning it to the above group, the best value is $t_i^+$, with the boundary convention stated above.  Therefore, over $S$, the objective is a constant plus
\[
\sum_{i\in S}b_i g_i,
\]
where $b_i\in\{0,1\}$ indicates whether index $i$ is assigned below.  The quantile restriction imposes
\[
\sum_{i\in S}b_i=K.
\]
This finite assignment problem is solved by setting $b_i=1$ for the
$K$ largest values of $g_i$.  If an alleged optimum assigned $i$ above and
$j$ below with $g_i>g_j$, swapping their assignments would strictly increase
the objective while preserving the count, a contradiction.
\end{proof}

\subsubsection{Bounded transformations}

\begin{proposition}[Bounded-transformation ceiling]
\label{prop:appendix_bounded}
Let $J=1$, and suppose Proposition~\ref{prop:appendix_lagrangian} applies.
Let $c^*(r)$ be a dual minimizer for the restriction
$\E[f(y^*)]=\kappa_f$.  Suppose that for some finite constant $B$,
\[
\sup_{t\in[y_L,y_U]}f(t)-\inf_{t\in[y_L,y_U]}f(t)
\leq B
\qquad
\PP\text{-a.s.}
\]
Then
\[
0
\leq
h_{\Theta^I}(r)-h_{\Theta^I_{y\mid f,\kappa_f}}(r)
\leq
|c^*(r)|B.
\]
\end{proposition}

\begin{proof}
The lower bound follows from
\[
\Sel^1(Y\mid f,\kappa_f)\subseteq\Sel^1(Y).
\]
Let $t_0(\omega)$ be a measurable maximizer of $s_r(\omega)t$ over
$t\in[y_L(\omega),y_U(\omega)]$.  Then
\[
h_{\Theta^I}(r)=\E[s_rt_0].
\]
Let $y'$ be any feasible selection satisfying
\[
\E[f(y')]=\kappa_f.
\]
Since $y'(\omega)$ and $t_0(\omega)$ both lie in the same interval
$[y_L(\omega),y_U(\omega)]$,
\[
|f(y'(\omega))-f(t_0(\omega))|\leq B
\qquad
\PP\text{-a.s.}
\]
Thus
\[
|\kappa_f-\E[f(t_0)]|
=
|\E[f(y')]-\E[f(t_0)]|
\leq B.
\]
Using the dual formula at $c^*=c^*(r)$,
\[
h_{\Theta^I_{y\mid f,\kappa_f}}(r)
=
c^*\kappa_f
+
\E\left[
\sup_{t\in[y_L,y_U]}\{s_rt-c^*f(t)\}
\right].
\]
The supremum is at least its value at $t_0$, hence
\[
h_{\Theta^I_{y\mid f,\kappa_f}}(r)
\geq
c^*\kappa_f+\E[s_rt_0-c^*f(t_0)]
=
h_{\Theta^I}(r)
+
c^*\{\kappa_f-\E[f(t_0)]\}.
\]
Rearranging and using the preceding bound gives
\[
h_{\Theta^I}(r)-h_{\Theta^I_{y\mid f,\kappa_f}}(r)
\leq
|c^*|\,|\kappa_f-\E[f(t_0)]|
\leq
|c^*|B.
\]
\end{proof}

This bound should be interpreted cautiously.  The bounded range of $f$
limits one direct channel through which the auxiliary moment can affect the
support function, but the multiplier $|c^*(r)|$ is endogenous and can be
large.  Thus bounded transformations are not uniformly weak without further
control of their shadow values.

\subsection{Proof of Proposition \ref{prop:conditional_fwl}}

\begin{proof}
Fix
\[
y\in\Sel^1(Y\mid\kappa(x_1)).
\]
The block normal equations are
\[
\begin{pmatrix}
\Sigma_{11}&\Sigma_{12}\\
\Sigma_{21}&\Sigma_{22}
\end{pmatrix}
\begin{pmatrix}
\theta_1\\
\theta_2
\end{pmatrix}
=
\begin{pmatrix}
\E[wy]\\
\E[x_2y]
\end{pmatrix}.
\]
The conditional restriction implies
\[
\E[wy]
=
\E[w\E[y\mid x_1]]
=
\E[w\kappa(x_1)]
=
g_1.
\]
Thus
\[
\theta_1
=
\Sigma_{11}^{-1}
(g_1-\Sigma_{12}\theta_2).
\]

Substitution into the second block gives
\[
\Sigma_{2\cdot1}\theta_2
=
\E[x_2y]
-
\Sigma_{21}\Sigma_{11}^{-1}g_1.
\]
Since $\E[wy]=g_1$,
\[
\E[x_2^*y]
=
\E[x_2y]
-
\Sigma_{21}\Sigma_{11}^{-1}g_1.
\]
Therefore,
\[
\theta_2
=
\Sigma_{2\cdot1}^{-1}\E[x_2^*y].
\]

This proves that every admissible coefficient belongs to the lifted set.

Conversely, take $\theta_2\in\Theta_2^I$. By definition, there exists
$y\in\Sel^1(Y\mid\kappa(x_1))$ such that
\[
\theta_2
=
\Sigma_{2\cdot1}^{-1}\E[x_2^*y].
\]
Define
\[
\theta_1
=
\Sigma_{11}^{-1}
(g_1-\Sigma_{12}\theta_2).
\]
Reversing the preceding algebra shows that the full vector satisfies the
normal equations for the same admissible selection $y$.

\end{proof}
It remains to justify the three claims stated in the main text immediately after the proposition. The map $\theta_2\mapsto(\theta_1',\theta_2')'$ with $\theta_1=\Sigma_{11}^{-1}(g_1-\Sigma_{12}\theta_2)$ is affine in $\theta_2$ and its second block returns $\theta_2$ itself, so distinct values of $\theta_2$ produce distinct images: the map is injective. Since $\Theta^I_{\kappa(x_1)}$ is therefore the image of $\Theta_2^I\subseteq\R^{d_2}$ under an injective affine map, its affine dimension equals that of $\Theta_2^I$, which is at most $d_2$. Finally, every point of $\Theta^I_{\kappa(x_1)}$ satisfies the $d_1+1$ scalar equations $\theta_1=\Sigma_{11}^{-1}(g_1-\Sigma_{12}\theta_2)$, one for each component of $\theta_1\in\R^{d_1+1}$; these are independent because $\Sigma_{11}$ is invertible, and they are exactly the $d_1+1$ exact linear restrictions referred to in the main text.
\subsection{Proof of Proposition \ref{prop:external_sharp}}

\begin{proof}
A coefficient vector is compatible with the external conditional restriction
if and only if there exists
\[
y\in\Sel^1(Y\mid\kappa(v))
\]
such that
\[
\E[\tilde x(y-\tilde x'\theta)]=0.
\]
Since $Q$ is nonsingular, this is equivalent to
\[
\theta=Q^{-1}\E[\tilde xy].
\]
Thus the set in \eqref{eq:external_sharp_set} is exactly the collection of
compatible coefficients and is sharp.

Since
\[
\Sel^1(Y\mid\kappa(v))\subseteq\Sel^1(Y),
\]
we also have
\[
\Theta^I_{\kappa(v)}\subseteq\Theta^I.
\]
Moreover, every $y \in \Sel^1(Y|\kappa(v))$ satisfies 
$$E[y]=E[E[y|v]]=E[\kappa(v)]=\kappa,$$
so $\Sel^1(Y|\kappa(v))\subseteq\Sel^1(Y|\kappa)$, which yields the claim in the Proposition.
\end{proof}

\subsection{Proof of Proposition \ref{prop:external_contraction}}

\begin{proof}
Under the conditional restriction, the allocated width mass in cell $m$ is
fixed at $A_m$. Indeed,
\[
A_m
=
p_m\{\kappa(v_m)-\E[y_L\mid v=v_m]\},
\]
while
\[
\mu_m(\Omega)
=
p_m\E[\Delta\mid v=v_m].
\]
Thus the conditional restriction in cell $m$ is exactly
\[
\int \tau\,d\mu_m=A_m.
\]
Therefore,
\[
h_{\Theta^I_{\kappa(v)}}(r)
=
\E[s_ry_L]
+
\sum_{m=1}^M
\mathcal T_{A_m}^{\mu_m}(s_r).
\]

Under the unconditional restriction, only the total width mass is fixed:
\[
\sum_{m=1}^M a_m=\alpha.
\]
For any feasible split $(a_1,\dots,a_M)$, allocation decisions may be made
independently within cells. Hence
\[
h_{\Theta^I_\kappa}(r)
=
\E[s_ry_L]
+
\max_{\substack{
0\leq a_m\leq\mu_m(\Omega)\\
\sum_m a_m=\alpha
}}
\sum_{m=1}^M
\mathcal T_{a_m}^{\mu_m}(s_r).
\]
The imposed split $(A_1,\dots,A_M)$ is feasible, proving nonnegativity of the
difference.

The pooled optimum is generated by a common threshold $c^*$ applied across
all cells, with cell-specific tie fractions when necessary. Therefore, the
conditional and unconditional support values coincide exactly when the
imposed budgets are generated by such a common threshold, that is, when
there are tie fractions $\gamma_m\in[0,1]$ such that
\[
A_m
=
\mu_m(s_r>c^*)
+
\gamma_m\mu_m(s_r=c^*)
\]
for every $m$.

If $s_r$ is independent of $v$ under the normalized width measure and the
budgets are proportional to cell masses,
\[
A_m
=
\frac{\mu_m(\Omega)}{\mu(\Omega)}\alpha,
\]
then the same quantile cutoff satisfies every cell constraint, and the
difference is zero.
\end{proof}

\subsection{Proof of Corollary \ref{cor:between_cell}}

\begin{proof}
For each cell define the dual objective
\[
\phi_m(c)
=
cA_m
+
\int(s_r-c)^+\,d\mu_m.
\]
Then
\[
\mathcal T_{A_m}^{\mu_m}(s_r)
=
\inf_c\phi_m(c),
\]
with minimizer $c_m^*$.

Define the pooled dual objective
\[
\Phi(c)
=
c\alpha
+
\sum_{m=1}^M
\int(s_r-c)^+\,d\mu_m.
\]
Since
\[
\sum_{m=1}^M A_m=\alpha,
\]
we have
\[
\Phi(c)=\sum_{m=1}^M\phi_m(c).
\]
The pooled cutoff $c^*$ minimizes $\Phi$, and therefore
\[
h_{\Theta^I_\kappa}(r)-\E[s_ry_L]
=
\Phi(c^*)
=
\sum_{m=1}^M\phi_m(c^*).
\]
Consequently,
\[
h_{\Theta^I_\kappa}(r)
-
h_{\Theta^I_{\kappa(v)}}(r)
=
\sum_{m=1}^M
\{\phi_m(c^*)-\phi_m(c_m^*)\}.
\]

Since $\phi_m'(c_m^*)=0$, Taylor's theorem with Lagrange remainder gives, for each $m$, a point $\xi_m$ between $c^*$ and $c_m^*$ such that
\[
\phi_m(c^*)-\phi_m(c_m^*)
=
\frac12\phi_m''(\xi_m)(c^*-c_m^*)^2
=
\frac12\mu_m(\Omega)f_m(\xi_m)(c_m^*-c^*)^2.
\]
Because $\xi_m$ lies between $c^*$ and $c_m^*$, $|\xi_m-c^*|\le|c_m^*-c^*|$, which tends to $0$ along the sequence by hypothesis. By the assumed uniform continuity of the densities $f_m$ in neighborhoods containing $c^*$ and every $c_m^*$ along the sequence,
\[
\max_{1\le m\le M}|f_m(\xi_m)-f_m(c^*)|\ \longrightarrow\ 0.
\]
Hence
\[
\phi_m(c^*)-\phi_m(c_m^*)
=
\frac12\mu_m(\Omega)f_m(c^*)(c_m^*-c^*)^2
+
\frac12\mu_m(\Omega)\bigl[f_m(\xi_m)-f_m(c^*)\bigr](c_m^*-c^*)^2,
\]
and summing the second term over the finitely many cells is $o\left(\sum_{m=1}^M\mu_m(\Omega)(c_m^*-c^*)^2\right)$, since the bracketed factor is uniformly small across all $m$ while $\mu_m(\Omega)$ and $M$ are fixed. Summing over the finitely many cells proves \eqref{eq:between_cell_expansion}.
\end{proof}

\section{Additional Numerical Details}
\label{app:numerical}

This appendix collects numerical detail supplementary to
Section~\ref{sec:numerical}: results that verify or illustrate the
theory of Section~\ref{sec:condexp} but are not central to the main
narrative. Code reproducing all calculations and figures in this appendix is available at
\url{https://github.com/BehroozMoosavi/Codes/tree/main/PI_with_%20Auxiliary_restriction}.

\subsection{Coordinate breadth}

Figure~\ref{fig:coord_bounds} plots the unconstrained coordinate bounds
for $\theta_0$ and $\theta_{\mathrm{educ}}$ (Proposition~\ref{prop:width}).
The coordinate-score identity $s_j=\tilde x_j^*/\sigma_j^2$ is verified
directly by computing the partialled-out residual $\tilde x_j^*$ from a
linear regression of $\mathrm{educ}$ on the intercept and comparing to the
closed-form score implied by $Q^{-1}$; the two agree to machine precision.

\begin{figure}[htbp]
\centering
\includegraphics[width=0.5\textwidth]{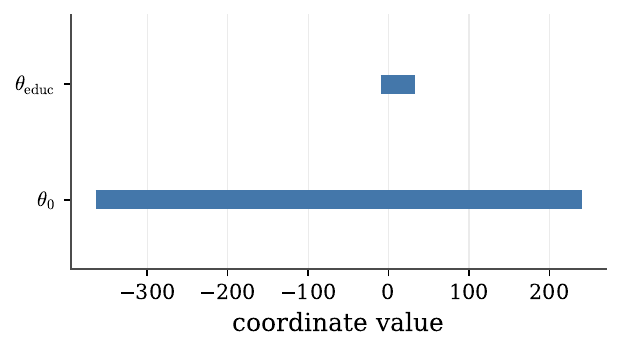}
\caption{Unconstrained coordinate bounds for $\theta_0$ and $\theta_{\mathrm{educ}}$.}
\label{fig:coord_bounds}
\end{figure}

\subsection{Local quadratic behavior}

Corollary~\ref{cor:local_mean_contraction} requires the score $s_r$ to
admit a continuous density under the width-weighted measure in a
neighborhood of zero. On the real CPS data this hypothesis fails outright: educational attainment takes only 12 distinct values, so $s_r$ has a discrete law under $\bar\mu$ with no density, and $f_r(0)$ in Corollary~\ref{cor:local_mean_contraction} is undefined rather than merely hard to estimate. The ratio plotted for the real-data curve therefore uses [SPECIFY: e.g., a kernel density estimate of $f_r$ evaluated at $0$] in place of $f_r(0)$, which is why the reported ratio diverges rather than converging to 1 as $h\to0$, exactly as expected once the density itself does not exist. Repeating the identical construction with a synthetic continuous
regressor in place of $\mathrm{educ}$ confirms the ratio converges cleanly
to 1 in that case, isolating discreteness of the data -- not an error in
the formula -- as the cause.

\begin{figure}[htbp]
\centering
\includegraphics[width=0.5\textwidth]{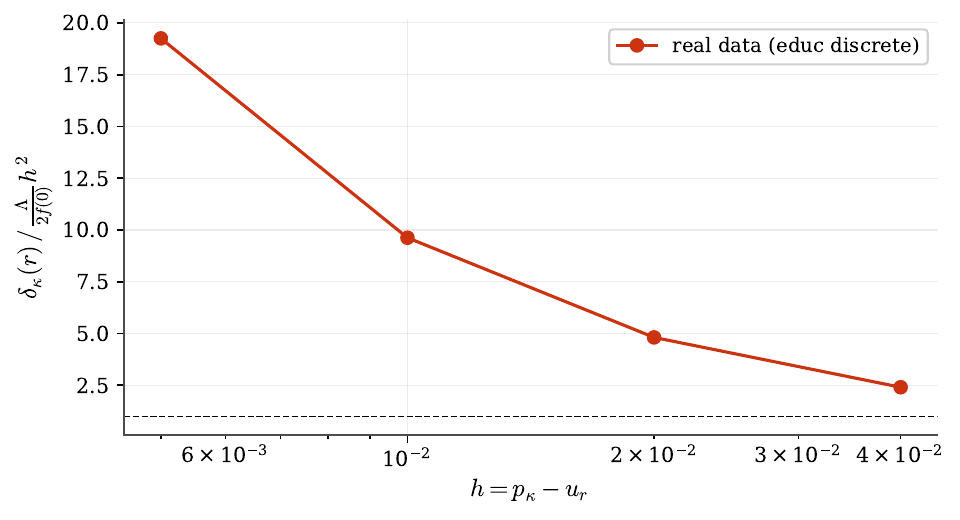}
\caption{Ratio of the realized contraction to its local quadratic
approximation, real (discrete) versus synthetic (continuous) regressor.}
\label{fig:local_quadratic}
\end{figure}

\subsection{The full three-dimensional picture for own-covariate conditioning}

Section~\ref{subsec:conditional_x1_num} reports only the resulting
one-dimensional coefficient interval for the return to education.
Figures~\ref{fig:case_c_3d} and~\ref{fig:case_c_projection} report the
full three-dimensional unconstrained region for
$\tilde x=(1,\mathrm{race},\mathrm{educ})$, together with the
conditional-on-race segment, and its exact two-dimensional projection
onto $(\theta_{\mathrm{race}},\theta_{\mathrm{educ}})$ obtained by
dropping the intercept coordinate from every vertex and re-computing the
convex hull -- exact, not approximate, since for a linear map $P$ and a
convex hull of vertices $V$, $P(\mathrm{conv}(V))=\mathrm{conv}(P(V))$.

The segment's endpoints lie exactly on the boundary of the full
three-dimensional $\Theta^I$ (verified analytically: the race-cell-specific
knapsack cutoffs are exactly reproducible as a threshold on a single
$\tilde x$-linear score, since race is itself a component of $\tilde x$).
This boundary-touching does not need to survive the projection to
$(\theta_{\mathrm{race}},\theta_{\mathrm{educ}})$ unless the supporting
hyperplane at those points has zero component in the dropped ($\theta_0$)
direction; here it does not, so the segment appears strictly interior in
Figure~\ref{fig:case_c_projection} despite touching the boundary in
Figure~\ref{fig:case_c_3d}.

\begin{figure}[htbp]
\centering
\includegraphics[width=0.6\textwidth]{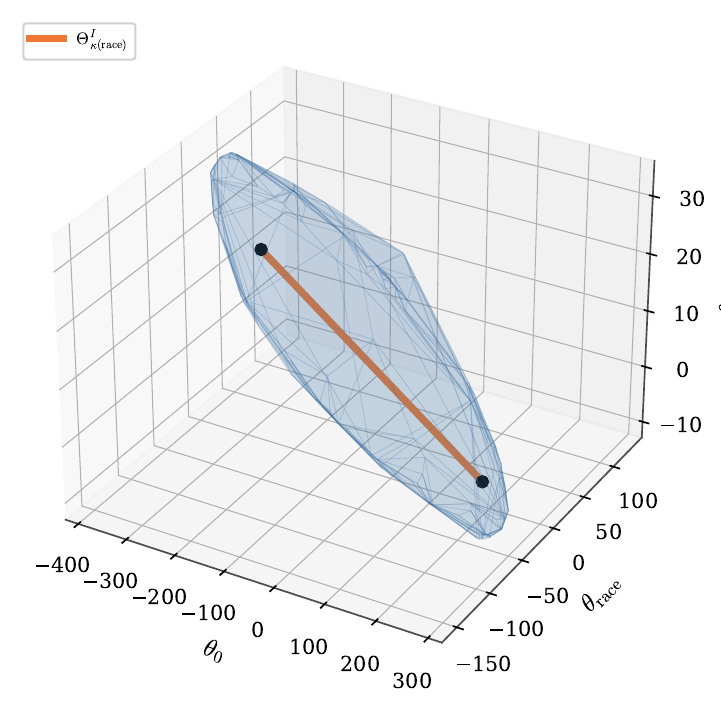}
\caption{Full three-dimensional unconstrained $\Theta^I$ for
$\tilde x=(1,\mathrm{race},\mathrm{educ})$, with the conditional-on-race
segment touching its boundary.}
\label{fig:case_c_3d}
\end{figure}

\begin{figure}[htbp]
\centering
\includegraphics[width=0.5\textwidth]{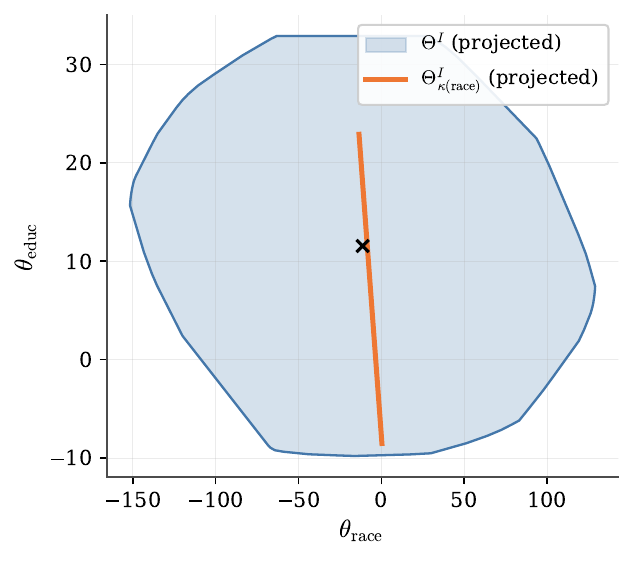}
\caption{Exact projection of Figure~\ref{fig:case_c_3d} onto
$(\theta_{\mathrm{race}},\theta_{\mathrm{educ}})$.}
\label{fig:case_c_projection}
\end{figure}

\subsection{Quantile information: a known median}

Figure~\ref{fig:median_quantile} illustrates
Proposition~\ref{prop:appendix_quantile} with $f(y)=\mathbf 1\{y\leq m\}$,
$m$ the sample median. The region $\Theta^I_{y\mid f,q}$ retains $96.5\%$ of the unconstrained area. Proposition~\ref{prop:appendix_bounded} bounds the contraction from above by $|c^*(r)|B$, and here $B=1$ regardless of the scale of income; but, as noted following that proposition, the multiplier $|c^*(r)|$ is endogenous and the bound need not be small in general. In this instance $|c^*(r)|$ is itself small, so the observed contraction is modest --- consistent with, but not implied uniformly by, boundedness of the transformation alone. The shrinkage is small enough that the inset zoom is needed to see it at all.

\begin{figure}[htbp]
\centering
\includegraphics[width=0.55\textwidth]{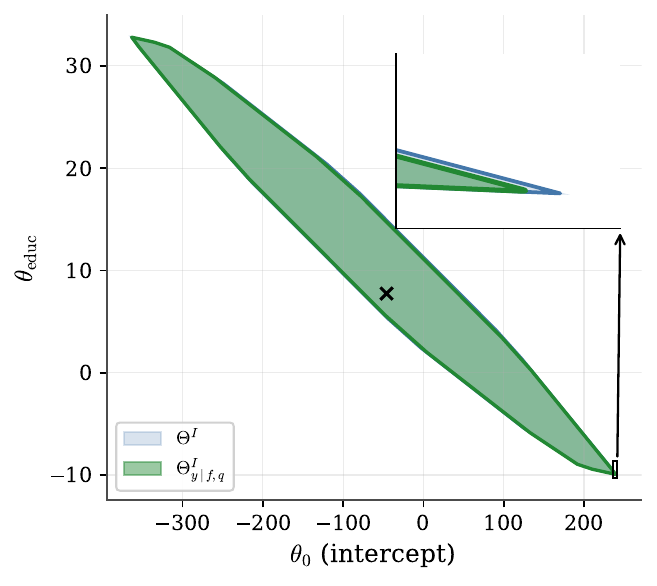}
\caption{$\Theta^I$ versus $\Theta^I_{y\mid f,q}$ under a known median,
with a zoomed inset showing the shrinkage.}
\label{fig:median_quantile}
\end{figure}

\subsection{Cutoff dispersion underlying the external-\texorpdfstring{$v$}{v} contraction}

Table~\ref{tab:cutoffs} reports the pooled cutoff $c^*$ and the five
age-cell-specific cutoffs $c_m^*$ underlying the extra contraction
reported in Section~\ref{subsec:conditional_v_num}, at the two coordinate
directions $r=e_{\mathrm{educ}}$ and $r=-e_{\mathrm{educ}}$ that determine,
respectively, the upper and lower endpoints of the $\theta_{\mathrm{educ}}$
interval in Figure~\ref{fig:external_v}. The two directions behave
differently. At $r=e_{\mathrm{educ}}$, every cell shares the pooled
cutoff exactly, so conditioning on age contributes no additional
contraction at the upper endpoint --- exactly why that endpoint is
unchanged between the pooled and conditional intervals in
Figure~\ref{fig:external_v}. At $r=-e_{\mathrm{educ}}$, two of the five
cells share the pooled cutoff but the remaining three do not, and it is
this dispersion alone that generates the positive extra contraction of
Proposition~\ref{prop:external_contraction} at the lower endpoint.

\begin{table}[htbp]
\centering
\caption{Pooled versus age-cell-specific cutoffs, $\theta_{\mathrm{educ}}$.}
\label{tab:cutoffs}
\begin{tabular}{lcc}
\toprule
& cutoff & differs from pooled? \\
\midrule
\multicolumn{3}{l}{\emph{$r=e_{\mathrm{educ}}$ (upper bound)}} \\
pooled ($c^*$) & 0.2562 & --- \\
cell 0 & 0.2562 & no \\
cell 1 & 0.2562 & no \\
cell 2 & 0.2562 & no \\
cell 3 & 0.2562 & no \\
cell 4 & 0.2562 & no \\
\midrule
\multicolumn{3}{l}{\emph{$r=-e_{\mathrm{educ}}$ (lower bound)}} \\
pooled ($c^*$) & 0.3157 & --- \\
cell 0 & 0.3157 & no \\
cell 1 & 0.3157 & no \\
cell 2 & 0.0298 & yes \\
cell 3 & 0.0298 & yes \\
cell 4 & 0.0298 & yes \\
\bottomrule
\end{tabular}
\end{table}

\end{document}